\documentclass[a4paper]{easychair}

\usepackage{latexsym,amssymb,amsmath,shadowbox,stmaryrd,url}
\usepackage{paralist,stackrel}
\usepackage[nompar,nosign]{commenting}
\declareauthor{lo}{Luke}{red}
\authorcommand{lo}{comment}
\declareauthor{sr}{Steven}{blue}
\authorcommand{sr}{comment}
\usepackage[square]{natbib}
\usepackage[all]{xy}
\usepackage[pgf]{pstring}
\usepackage{eufrak}
\usepackage[T1]{fontenc}
\usepackage{hyperref}
\hypersetup{
  pdfauthor={Luke Ong},
  pdftitle={Normalisation by Traversals},
  pdfsubject={Game Semantics},
  colorlinks,
  linkcolor=[rgb]{0,0.75,1},
  urlcolor=[rgb]{0,1,1},
  citecolor=[rgb]{1,0.5,1}
}
\usepackage{amsthm}
\newtheorem{theorem}{Theorem}[section] 
\newtheorem{lemma}[theorem]{Lemma} 
\newtheorem{corollary}[theorem]{Corollary}

\theoremstyle{definition} 
\newtheorem{definition}[theorem]{Definition}
\newtheorem{example}[theorem]{Example}

\theoremstyle{remark} 
\newtheorem{remark}[theorem]{Remark}

\newcommand\ong[1]{{#1}}

\newcommand\commentout[1]{}

\renewcommand\phi{\varphi}

\newcommand\dom{\mathit{dom}}
\newcommand\Var{\mathit{Var}}

\renewcommand\phi{\varphi}

\newcommand\arity{\mathit{ar}}

\newcommand\trmng[1]{{\mathfrak{Trav} \lform{#1}}}

\newcommand\lform[1]{\langle{#1}\rangle}

\newcommand\intseq[1]{\textbf{IntSeq}(#1)}
\newcommand\intseqpv[1]{\textbf{IntSeq}^{{\rm PV}}(#1)}

\newcommand\trunc[2]{{#1}_{{\leq}{#2}}}

\newcommand\calG{{\cal G}}
\newcommand\calL{{\cal L}}
\newcommand\calI{{\cal I}}
\newcommand\comp[2]{\mathit{Comp}_{#2}(#1)}
\newcommand\moves[1]{|#1|}
\newcommand\Omoves[1]{|#1|^{\OO}}
\newcommand\Pmoves[1]{|#1|^{\PP}}

\newcommand\natnum{{\mathbb N}}

\newcommand\pto{\rightharpoonup}

\newcommand\LamAt{\hbox{\boldmath $\Lambda$}(@)}
\newcommand\Lam{\hbox{\boldmath $\Lambda$}}
\newcommand\roundbra[1]{(#1)}
\newcommand\FV{\hbox{FV}}
\newcommand\Ar{\mathit{Ar}}
\newcommand\To{\Rightarrow}
\newcommand\Initial{\mathit{Init}}
\newcommand\Tree[1]{\mathit{Tree}_{#1}}
\newcommand\ArExp{{\mathit{ExpAr}}}
\newcommand\ArSuc{{\mathit{SucAr}}}

\newcommand\PP{\mathrm{P}}
\newcommand\OO{\mathrm{O}}
\newcommand\sdot{\cdot}

        \newcommand   \myendproof{{
           \parfillskip=0pt            
           \widowpenalty=10000         
               \displaywidowpenalty=10000  
               \finalhyphendemerits=0      
        \leavevmode                 
           \unskip                     
          \nobreak                    
                \hfil                       
        \penalty50                  
        \hskip2pt                   
        \null                       
        \hfill                      
         $\square$                   
        \par                        
        \penalty-200                
        \smallskip                  
          }
         }

        \newcommand\makeset[1]{\{\,#1\,\}}
        \newcommand\anglebra[1]{\langle\, #1 \,\rangle}

        \newcommand\mor{\longrightarrow}

        \newcommand\lub{\bigsqcup}

        \newcommand\proj[2]{{#1 \restriction #2}}

        \newcommand\pview[1]{\ulcorner{#1}\urcorner}
        \newcommand\oview[1]{\llcorner{#1}\lrcorner}

        \newcommand  \bohm{\mathrel{\lower.2ex
                \hbox{${\stackrel{\sqsubset}{\scriptscriptstyle \sim}}$}}}

        \newcommand\funsp{\Rightarrow}

        \newcommand\mng[1]{{\mathopen{[\![}\,#1\,\mathclose{]\!]}}}

\newcommand\blambda{\lambda}
        
        \newcommand\lterm[2]{{\blambda{#1}.{#2}}}

        \newcommand\seq[2]{{{#1}}\vdash{{#2}}}

        \newcommand\arrow[1]{\rightarrow_{#1}}

\title{Normalisation by Traversals}
\titlerunning{Normalisation by Traversals}

\author{C.-H.~Luke~Ong}
\authorrunning{C.-H.~Luke~Ong}

\institute{University of Oxford}

\begin{document}

   \maketitle

\begin{abstract}
We present a novel method of computing the $\beta$-normal $\eta$-long form of a simply-typed $\lambda$-term by constructing \emph{traversals} over a variant abstract syntax tree of the term.
In contrast to $\beta$-reduction, which changes the term by substitution, this method of normalisation by traversals leaves the original term intact. 
We prove the correctness of the normalisation procedure by game semantics. 
As an application, we establish a path-traversal correspondence theorem which is the basis of a key decidability result in higher-order model checking.

\end{abstract}


\tableofcontents

\sr{From Neil: 1. Use (FVar) and (BVar) for (Var) and (CC) instead.

2. Emphasise (BVar) is copy-cat, and the transfer of control between caller and callee.

3. Consider redefining directionality of @ to start from 0.

4. The (Lam) rule needs further work / clarification.}

\lo{LO: All done. I have added a discussion after Definition~\ref{def:traversal} (highlighted in red) to clarify 3 and 4.}

\sr{From Steven: I think the introduction reads very well.  The only suggestion I would make is regarding the presentation of the main theorem.  When the main theorem is presented at the top of page 4, you have already spent some effort in introducing traversals, but you have not introduced P-views.  Then the theorem relates root-justified traversals to maximal P-views over the denotation.  I think if the reader understands games semantics then this is fine, but if you want to reach a larger audience then it will be worthwhile spending a short paragraph introducing these maximal P-views. 

\lo{I have added an informal explanation of P-views. I have rewritten large parts of Section~\ref{sec:traversal-intro}, including a different and hopefully more detailed and helpful example of long form, more examples of traversals, and explanation of strong bijection between projected traversals and paths (= P-views) in the AST of the normal form.} 

After stating the theorem you say that they correspond to maximal paths in the abstract syntax tree of the normal form of the term -- I think this could just be more prominent in a paragraph of its own, perhaps with an example. The reason is that, if the reader is not familiar with P-views then they may just read it as ``there are two notions that you have never heard of and it is shown that they are strongly related''.  On the other hand, my feeling is that if they get to the point of reading the statement of the main theorem and they have already equated ``maximal P-view over the denotation'' with beta-normal form, then the utility of the theorem will be immediately clear.  Hence, I would spend a short paragraph ensuring that this connection is made in their mind. \lo{LO: Done. I have rewritten this part and presented it as a sketch of the correctness proof of normalisation by traversals.}

One small point about the definition of traversal on page 15: in the rule for (Lam) the last sentence ends ``if $n$ is labelled by a variable (as opposed to @) then its pointer is determined by the condition that 
$ t \cdot \lambda\overline\xi \cdot n$ is a path in $T$''.  
I was discussing this with Neil Jones because he was having some trouble with the definition.  
Could it be that you mean to say that ``if $n$ is labelled by a variable (as opposed to @) then its pointer is determined by the condition that $t \cdot \lambda\overline\xi \cdot n$ is a justified sequence''?  
I ask this because the condition of being a path in the tree does not seem to involve pointers. 
\lo{LO: You are right: I was not clear enough. A path in (the abstract syntax tree of) a long form \emph{is} a justified sequence; in fact, it is necessarily a P-view. This is clarified in Remark~\ref{rem:pview-path}.} 
Rather, the condition of being a path in the tree determines the label of $n$ and the pointer from $n$ is determined by the condition that the expression is a valid justified sequence.  
\lo{LO: I have added comments and examples after Definition~\ref{def:traversal}.}
Incidentally, I think the two occurrences of $T_B$ in this rule should both be $T_M$. \lo{LO: Yes, thanks.}}

\section{Introduction}

This paper is about a method of computing the normal form of a lambda-term by traversing a slightly souped up version of the abstract syntax tree of the term, called its \emph{long form}.
A \emph{traversal} is a certain \emph{justified sequence} of nodes of the tree i.e.~sequence of nodes such that each (non-initial) node is equipped with a \emph{justification pointer} to an earlier node.
Each traversal may be viewed as computing a path in the abstract syntax tree of the $\beta$-normal $\eta$-long form of the term.
Note that a term-tree, such as the normal form of a term, is determined by the set of its paths.
The usual (normalisation by) $\beta$-reduction changes the term by substitution.  
By contrast, our method of normalisation by traversals does not perform $\beta$-reduction, thus leaving the original term intact. 
In this sense, normalisation by traversals uses
a form of reduction that is \emph{non-destructive} and \emph{local} \citep{DanosR93}.


\subsection{Traversals: an example}
\label{sec:traversal-intro}

We first illustrate traversals with an example.
Take the term-in-context, 
\[
\seq{g : (o \to o) \to o \to o, a : o}{N \, P \, R : o},
\] 
where 
\begin{align*}
N &= \lterm{\phi^{(o \to o) \to o \to o} \, z^{o \to o}}{\phi \, (\lterm{x}{\phi \, (\lterm{x'}{x}) \, a}) \, (z \, a)}\\
P &= \lterm{f^{o \to o} \, y^o}{f \, (g \, (\lterm{b}{b})\, y)}\\
R &= g \, (\lterm{b'}{b'})
\end{align*}
which has normal form $g\,(\lterm{b}{b})\,(g \, (\lterm{b'}{b'}) \, a)$.
To normalise the term $N \, P \, R$ by traversal, we first construct its long form, written $\lform{N \, P \, R}$, which is the following term
\[
\lterm{}{@ \, 
\Big(
\lterm{\phi z}{\phi \, 
\big(
\lterm{x}{\phi \, (\lterm{x'}{x}) \, (\lterm{}{a})}
\big) 
\,
\big(
\lterm{}{z \, (\lterm{}{a})}
\big)
}
\Big) 
\, 
\Big(
\lterm{fy}{
f 
\big(
\lterm{}{g (\lterm{b}{b}) (\lterm{}{y})}
\big)
}
\Big) 
\, 
\Big(
\lterm{w}{g (\lterm{b'}{b'}) (\lterm{}{w})}}
\Big)
\]
The long form is obtained by $\eta$-expanding the term fully\footnote{Somewhat nonstandardly, every ground-type subterm that is \emph{not} in a function position is also expanded (to a term with a ``dummy lambda'') $t \mapsto \lterm{}{t}$. For example, $\lterm{x^{o \to o \to o}}{x \, a}$ fully $\eta$-expands to $\lterm{x^{o \to o \to o} z^o}{x \, (\lterm{}{a}) \, (\lterm{}{z})}$; 
and $g\,(\lterm{b}{b})\,(g \, (\lterm{b'}{b'}) \, a)$ fully $\eta$-expands to 
$\lterm{}{g\,(\lterm{b}{b})\,(\lterm{}{g \, (\lterm{b'}{b'}) \, \lterm{}{a}})}$.}, 
and then replacing the (implicit) binary application operator of each redex by the \emph{long application} operator $@$.

Now consider the abstract syntax tree of $\lform{N \, P \, R}$, as shown in Figure~\ref{fig:traversaleg}.
Notice that nodes on levels 0, 2, 4, etc., are labelled by lambdas; and those on levels 1, 3, 5, etc., are labelled by either variables or the long application symbol $@$.
The dotted arrows (pointing from a variable to its $\lambda$-binder) indicate an enabling relation between nodes of the tree: $n \vdash n'$ (read ``$n'$ is enabled by $n'$'') just if \(\xymatrix@C=.7cm{n & n' \ar@{.>}@/_.5pc/[l]}\).
By convention, (nodes labelled by) free variables are enabled by the root node, as indicated by the dotted arrows.
Further, every lambda-labelled node, except the root, is enabled by its parent node in the tree (we omit all such dotted arrows from the figure to avoid clutter).

Traversals are \emph{justified sequences} (i.e., sequences of nodes whereby each (non-initial) node has a justification pointer to an earlier node) that strictly alternate between lambda and non-lambda labels. 
The long form $\lform{N \, P \, R}$ has three maximal traversals, one of which is the following:
\begin{equation}
\small
\Pstr[0.65cm]{
(n1){\blambda}\,
(n2){@}\,
(n3-n2,33:1){\blambda \phi z}\,
(n4-n3){\phi}\,
(n5-n2,35:2){\stackrel[5]{}{\blambda f y}}\,
(n6-n5){f}\,
(n7-n4,38:1){\blambda x}\,
(n8-n3){\phi}\,
(n9-n2,43:2){\blambda f y}\,
(n10-n9){\stackrel[10]{}{f}}\,
(n11-n8,38:1){\blambda x'}\,
(n12-n7){x}\,
(n13-n6,45:1){\blambda}\,
(n14-n1){g}\,
(n15-n14,30:2){\stackrel[15]{}{\blambda}}\,
(n16-n5){y}\,
(n17-n4,45:2){\blambda}\,
(n18-n3){z}\,
(n19-n2,45:3){\blambda w}\,
(n20-n1){\stackrel[20]{}{g}}\,
(n21-n20,30:1){\blambda b'}\,
(n22-n21){b'}\,
}
\label{eq:traversal-intro-1}
\end{equation}

The five rules that define traversals are displayed in Table~\ref{tab:traversals}.
The rule (Root) says that the root node is a traversal. 
The rule (Lam) says if a traversal $t$ ends in a $\blambda$-labelled node $n$, then $t$ extended with the child node of $n$, $n'$, is also a traversal. 
Note that every node of a traversal in an even position is constructed by rule (Lam).
The rule (App) justifies the construction of the third node of traversal (\ref{eq:traversal-intro-1}).
If a traversal ends in a node labelled with a variable $\xi_i$,
then there are two cases, corresponding to whether $\xi_i$ is (hereditarily justified\footnote{We say that a node-occurrence $n$ in a justified sequence is hereditarily justified by another $n'$ if there is a chain of pointers from $n$ to $n'$.} by) a bound (BVar) or free (FVar) variable in the long form.

(BVar): If the traversal has the form
$\Pstr[0.7cm]{t \sdot (n){n} \sdot (lxi){\blambda \overline\xi} \cdots (xi-lxi,45){\xi_i} }$ where $\overline \xi = \xi_1 \cdots \xi_n$, and $\xi_i$ is hereditarily justified by a $@$, then there are two subcases. 
\begin{itemize}[-] 
\item If $n$ is (labelled by) a variable then $\Pstr[0.7cm]{t \sdot (n){n} \sdot (lxi){\blambda \overline\xi} \cdots (xi-lxi,45){\xi_i} \sdot (leta-n,45:i){\blambda\overline\eta}}$ is a traversal, whereby the pointer label $i$ means that the node $\blambda \overline \eta$ is the $i$-th child of $n$. 
For example, this rule justifies the construction of the 11th node $\blambda x'$ and 13th node $\blambda$ of  traversal (\ref{eq:traversal-intro-1}).

\item If $n$ is (labelled by) $@$ then $\Pstr[0.7cm]{t \sdot (n){@} \sdot (lxi){\blambda \overline\xi} \cdots (xi-lxi,45){\xi_i} \sdot (leta-n,45:{i+1}){\blambda\overline\eta}}$ is a traversal. 
For example, this rule justifies the construction of the 5th node $\blambda f y$, 9th node $\blambda f y$ and 19th node $\blambda w$ of traversal (\ref{eq:traversal-intro-1}). 
meaning that the node is the $(i+1)$-th child of $@$.
\end{itemize}
Intuitively the rule (BVar) captures the switching of control between caller and callee, or between formal and actual parameters. 
See Remark~\ref{rem:traversals} for further details.

(FVar): If the traversal has the form $\Pstr[0.7cm]{t \sdot (lxi){\blambda \overline\xi} \cdots (xi-lxi,45){\xi_i}}$ and $\xi_i$ is hereditarily justified by the opening node $\epsilon$, then
$\Pstr[0.7cm]{t \sdot (lxi){\blambda \overline\xi} \cdots (xi-lxi,45){\xi_i} \sdot (leta-xi,35:j){\blambda\overline\eta}}$ is a traversal, for each child-node $\blambda\overline\eta$ of $\xi_i$ (so $j$ ranges over $\makeset{1, \cdots, \arity(\xi_i)}$ where $\arity(\xi_i)$ is the arity (branching factor) of $\xi_i$). 
For example, the 15th node $\blambda$ and the 21st node $\blambda$ of traversal (\ref{eq:traversal-intro-1}) are constructed by this rule.

As mentioned earlier, each traversal computes a path in the abstract syntax tree of the $\beta$-normal $\eta$-long form of the term $N \, P \, R$, which is shown in Figure~\ref{fig:nf-NPR}. 
With reference to the Figure, notice that each path of the tree is actually an alternating justified sequence, in fact, a {P-view} (about which more anon). 
Such a justified path is obtained from the traversal by projecting to those nodes that are hereditarily justified by the root node.
Thus we obtain the following projected justified subsequence from traversal (\ref{eq:traversal-intro-1}), 
\(
\Pstr[0.65cm]{
(n1){\blambda}\, 
(n2-n1){g}\, 
(n3-n2,30:2){\blambda}\, 
(n4-n1){g}\, 
(n5-n4,30:1){\blambda b'}\, 
(n6-n5){b'}
},
\) 
which is a maximal path of 
$\lterm{}{
g \, (\lterm{b}{b}) \, 
(
\lterm{}{g \, (\lterm{b'}{b'}) \, (\lterm{}{a})}
)
}$,
the $\beta$-normal $\eta$-long form of $N \, P \, R$.

The other maximal traversals of $\lform{N \, P \, R}$ are:
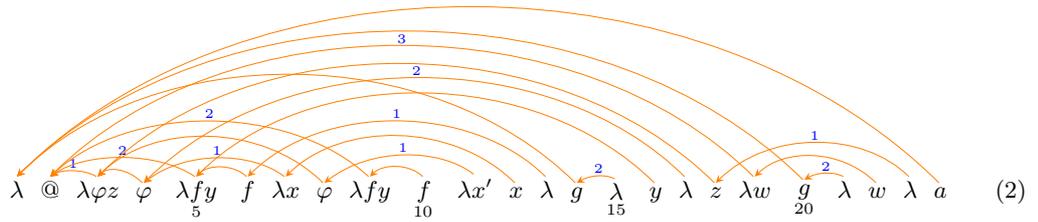
\begin{equation}
\small
\Pstr[0.65cm]{
(n1){\blambda}\,
(n2){@}\,
(n3-n2,33:1){\blambda \phi z}\,
(n4-n3){\phi}\,
(n5-n2,35:2){\stackrel[5]{}{\blambda f y}}\,
(n6-n5){f}\,
(n7-n4,38:1){\blambda x}\,
(n8-n3){\phi}\,
(n9-n2,43:2){\blambda f y}\,
(n10-n9){\stackrel[10]{}{f}}\,
(n11-n8,38:1){\blambda x'}\,
(n12-n7){x}\,
(n13-n6,45:1){\blambda}\,
(n14-n1){g}\,
(n15-n14,30:2){\stackrel[15]{}{\blambda}}\,
(n16-n5){y}\,
(n17-n4,45:2){\blambda}\,
(n18-n3){z}\,
(n19-n2,45:3){\blambda w}\,
(n20-n1){\stackrel[20]{}{g}}\,
(n21-n20,30:2){\blambda}\,
(n22-n19){w}\,
(n23-n18,45:1){\blambda}\,
(n24-n1){a}
}
\label{eq:traversal-intro-2}
\end{equation}
\begin{equation}
\small
\Pstr[0.65cm]{
(n1){\blambda}\,
(n2){@}\,
(n3-n2,33:1){\blambda \phi z}\,
(n4-n3){\phi}\,
(n5-n2,35:2){\stackrel[5]{}{\blambda f y}}\,
(n6-n5){f}\,
(n7-n4,38:1){\blambda x}\,
(n8-n3){\phi}\,
(n9-n2,43:2){\blambda f y}\,
(n10-n9){\stackrel[10]{}{f}}\,
(n11-n8,38:1){\blambda x'}\,
(n12-n7){x}\,
(n13-n6,45:1){\blambda}\,
(n14-n1){g}\,
(n15-n14,30:1){\stackrel[15]{}{\blambda b}}\,
(n16-n15){b}
}
\label{eq:traversal-intro-3}
\end{equation}
which respectively project to the maximal paths, 
\(
\Pstr[0.65cm]{
(n1){\blambda}\, 
(n2-n1){g}\, 
(n3-n2,30:2){\blambda}\, 
(n4-n1){g}\, 
(n5-n4,30:2){\blambda}\, 
(n6-n1){a}
}
\)
and
\(
\Pstr[0.65cm]{
(n1){\blambda}\, 
(n2-n1){g}\, 
(n3-n2,30:1){\blambda b}\, 
(n4-n3){b}
}
\)
, of 
the $\beta$-normal $\eta$-long form of $N \, P \, R$.

\begin{figure}
\[\xymatrix@R-.4cm@C-.1cm{
 & & & \blambda \ar@{-}[d] & & & & & \\ 
 & & & @ \ar@{-}[dl] \ar@{-}[drr] \ar@{-}[drrrr] & & & & & \\ 
 & & \blambda \phi z \ar@{-}[d] & & & \blambda f y \ar@{-}[d] & & \blambda w \ar@{-}[d] & \\ 
 & & \phi \ar@{-}[dl] \ar@{-}[dr] \ar@{.>}@/^1pc/[u] & & & f \ar@{-}[d] \ar@{.>}@/_.7pc/[u] & & g \ar@{-}[ld] \ar@{-}[rd] \ar@{.>}@/_1pc/[uuullll] & \\ 
 & \blambda x \ar@{-}[d] & & \blambda \ar@{-}[d] & & \blambda \ar@{-}[d] & \blambda b' \ar@{-}[d] & & \blambda \ar@{-}[d] \\ 
 & \phi \ar@{-}[dl] \ar@{-}[dr] \ar@{.>}@/_.5pc/[uuur] & & z \ar@{-}[d] \ar@{.>}@/^1pc/[uuul] & & g \ar@{.>}@/^1pc/[uuuuull] \ar@{-}[ld] \ar@{-}[rd] & b' \ar@{.>}@/_.7pc/[u] & & w \ar@{.>}@/_2pc/[uuul]\\ 
\blambda x' \ar@{-}[d]  & & \blambda \ar@{-}[d] & \blambda \ar@{-}[d] & \blambda b\ar@{-}[d] & & \blambda \ar@{-}[d] & & \\ 
x \ar@{.>}@/^2.5pc/[uuur] & & a \ar@{.>}@/^1.5pc/[uuuuuuur] & a \ar@{.>}@/_2pc/[uuuuuuu] & b \ar@{.>}@/_.7pc/[u] & & y \ar@{.>}@/_2pc/[uuuuul] & &}
\]
\caption{The abstract syntax tree of the long form $\lform{N \, P \, R}$ in Section~\ref{sec:traversal-intro}.\label{fig:traversaleg}}
\end{figure}

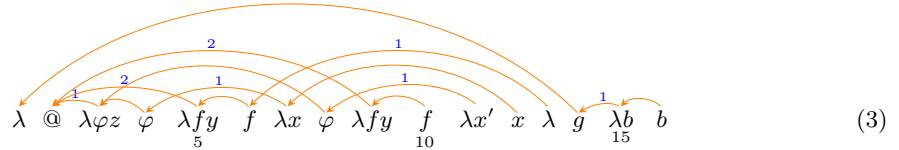
\begin{figure}
\[
\xymatrix@R-.4cm@C-.1cm{
& \blambda \ar@{-}[d] & & \\
& g \ar@{-}[ld] \ar@{-}[rd] \ar@{.>}@/_.7pc/[u] & & \\
\blambda b \ar@{-}[d] & & \blambda \ar@{-}[d] & \\
b \ar@{.>}@/_.7pc/[u] & & g \ar@{-}[ld] \ar@{-}[rd] \ar@{.>}@/_2pc/[uuul] & \\
& \blambda b' \ar@{-}[d] & & \blambda \ar@{-}[d]\\
& b' \ar@{.>}@/_.7pc/[u] & & a \ar@{.>}@/_3.5pc/[uuuuull]
}
\]
\caption{The abstract syntax tree of $\lform{g \, (\lterm{b}{b}) \, (g \, (\lterm{b'}{b'}) \, a)}$ in Section~\ref{sec:traversal-intro}.\label{fig:nf-NPR}}
\end{figure}

\subsection{Correctness of normalisation by traversals}

We state the correctness theorem.

\begin{theorem}[Correctness]
Given a term-in-context $\seq{\Gamma}{M : A}$, there is a bijection between the following sets of justified sequences:
\begin{itemize}[-]
\item $\proj{\trmng{M}}{\epsilon}$, traversals over $\lform{M}$ projected to nodes hereditarily justified by the root node $\epsilon$
\item $\mathit{Path}\lform{\beta(M)}$, justified paths in the abstract syntax tree of the long form $\lform{\beta(M)}$, where $\beta(M)$ is the $\beta$-normal form of $M$. 
\end{itemize}
Furthermore the bijective map is \emph{strong} in the sense that every projected traversal and its image as path are isomorphic as justified sequences.
Thus, normalisation by traversals is correct: the set of projected traversals over $\lform{M}$ determine the $\beta$-normal $\eta$-long form of $M$.
\end{theorem}

The proof is via game semantics \citep{HO00}.
The game-semantic denotation of a sequent, $\mng{\seq{\Gamma}{M : A}}$, is an \emph{innocent} strategy, represented as a certain prefix-closed set of justified sequences called \emph{plays}.
Innocent means that the strategy is generated by the subset $\pview{\mng{\seq{\Gamma}{M : A}}}$ of plays which are \emph{P-views}. 
(Intuitively the P-view of a play is a certain justified subsequence consisting only of those moves which player P considers relevant for determining his next move. 
See Section~\ref{sec:games-basics} for the definitions.)
We prove the correctness theorem by showing that the following three sets of justified sequences are strongly bijective:

\begin{equation}
\xymatrix@C=1cm{
\proj{\trmng{M}}{\epsilon} \ar@{->}[r]^{\hat{\ell}^\ast} & 
\pview{\mng{\seq{\Gamma}{M : A}}} \ar@{->}[r]^{{\cal G}} & \mathit{Path}\lform{\beta(M)}
}
\label{eq:mainresult}
\end{equation}


The strong bijection on the right, $\cal G$, is a well-known fundamental result of game semantics, and the essence of the definability result \citep{HO00}.
The strong bijection on the left, $\hat{\ell}^\ast$, is the main technical result of the paper, Theorem~\ref{thm:corr}.
The key intuition is that traversals correspond to a certain collection of \emph{uncovered plays} (of an innocent strategy) or plays-without-hiding.
Given a term-in-context $\seq{\Gamma}{M : A}$, we formalise the long form, $\lform{M}$, as a $\Sigma$-labelled binding tree, in the sense of \cite{Stirling09}.
We then identify two arenas associated with $\lform{M}$, viz., \emph{explicit arena} $\ArExp\lform{M}$, and \emph{succinct arena} $\ArSuc\lform{M}$.
The enabling relation $\vdash$ between nodes of the long form $\lform{M}$ (as discussed in Section~\ref{sec:traversal-intro}) is defined as the enabling relation of the arena $\ArExp\lform{M}$, whose underlying set consists of nodes of the (binding) tree $\lform{M}$. 
Traversals over $\lform{M}$ are then defined as justified sequences over $\ArExp\lform{M}$ by induction over a number of rules.

To interpret traversals over $\lform{M}$ as uncovered plays, we define the succinct arena $\ArSuc\lform{M}$ as a disjoint union of 
\begin{itemize}[-]
\item a ``revealed'' arena, consisting of the arena over which $\mng{\seq{\Gamma}{M : A}}$ is defined as a strategy, and 
\item a ``hidden'' arena, for interpreting the moves that are hereditarily justified by an $@$ in $\lform{M}$.
\end{itemize}

We then define a map $\hat\ell : \ArExp\lform{M} \to \ArSuc\lform{M}$, called \emph{direct arena morphism}, which preserves initial moves, and preserves and reflects the enabling relation.
Furthermore the morphism extends to a function $\hat\ell^\ast$ that maps justified sequences of $\ArExp\lform{M}$ to those of $\ArSuc\lform{M}$.
Theorem~\ref{thm:corr} then asserts that the map $\hat\ell^\ast$ defines a strong bijection from $\proj{\trmng{M}}{\epsilon}$ to $\pview{\mng{\seq{\Gamma}{M : A}}}$.

\subsection*{Application to higher-order model checking} 

We apply normalisation by traversals to higher-order model checking \citep{Ong15}.
The Higher-order Model Checking Problem \citep{KnapikNU02} asks, given a higher-order recursion scheme $\calG$ and a monadic second-order formula $\phi$, 
whether the tree generated by $G$, written $\mng{\calG}$, satisfies $\phi$?

This problem was first shown to be decidable by \cite{Ong06}.
Ong's proof uses a \emph{transference principle}: instead of reasoning about the parity winning condition of infinite paths in the generated tree $\mng{\calG}$, 
he considers \emph{traversals} over the computation tree of $\calG$, $\lambda(\calG)$, which is a tree obtained from $\calG$ by first transforming the rewrite rules into long forms, 
then unfolding these transformed rules \emph{ad infinitum}, but without performing any $\beta$-reduction 
(i.e.~substitution of actual parameters for formal parameters).
The argument uses a key technical lemma, presented in the following as Theorem~\ref{thm:11traversal}, which 
states that paths in the generated tree $\mng{\calG}$ on the one hand, and {traversals} over the computation tree $\lambda(\calG)$ projected to the terminal symbols from $\Sigma$ on the other, are the same set of finite and infinite sequences over $\Sigma$.
In this paper, we apply Theorem~\ref{thm:corr} to prove Theorem~\ref{thm:11traversal}.




\section{Technical preliminaries}

We write $\natnum = \makeset{1, 2, \cdots}$, $X_0 = \makeset{0} \cup X$ for $X \subseteq \natnum$, $[n] = \makeset{1, \cdots, n}$ for $n \in \natnum$, $X^\star$ for the set of finite sequences of elements of $X$, $\leq$ for the prefix ordering over sequences, and $|x_1 \sdot x_2 \cdots x_n| = n$ for the length of sequences.
By a \emph{tree} $T$, we mean a subset of $\natnum^\star$ that is prefix-closed (i.e.~if $\alpha \in T$ and $\alpha' \leq \alpha$ then $\alpha' \in T$) and order-closed (i.e.~if $\alpha \sdot n \in T$ and $1 \leq n' < n$ then $\alpha \sdot n' \in T$).
A \emph{path} in $T$ is a sequence of elements of $T$, $\alpha_1 \sdot \alpha_2 \cdots \alpha_n$, such that $\alpha_1 = \epsilon$, $\alpha_n = i_1 \sdot i_2 \cdots i_{n-1}$, and $\alpha_{j+1} = \alpha_j \sdot i_{j}$ for each $1 \leq j \leq n-1$.
We write $\mathit{Path}(T)$ for the set of paths in the tree $T$.
A \emph{ranked alphabet} $\Sigma$ is a set of symbols such that each symbol $f \in \Sigma$ has an arity $\arity(f) \geq 0$.
A \emph{$\Sigma$-labelled tree} is a function $F : T \to \Sigma$ such that 
\begin{inparaenum}[(i)]
\item $T$ is a tree, and
\item for each $\alpha \in T$, if $\arity(F(\alpha)) = n$ then $\makeset{1, \cdots, n} = \makeset{i \mid \alpha \sdot i \in T}$.
\end{inparaenum}
By definition, $\Sigma$-labelled trees are \emph{ordered}, i.e., the set of children of each node is a (finite) linear order.
Let $T$ be a tree, and let $\alpha \in T$. The tree \emph{$T$ rooted at $\alpha$}, denoted $T_{@\alpha}$, is the set $\makeset{\gamma \in \natnum^\ast \mid \alpha \sdot \gamma \in T}$. 

\emph{Types} (ranged over by $A, B$, etc.) are defined by the grammar: $A ::= o \; | \; (A \to B)$. 
A type can be written uniquely as (by convention $\to$ associates to the right), $A_1 \to \cdots \to A_n \to o$, which we abbreviate to $(A_1, \cdots, A_n, o)$.
The \emph{order} of a type $A$, $\mathit{ord}(A)$, which measures how deeply nested a type is on the left of the arrow, is defined as 
$\mathit{ord}(o) := 0$, and $\mathit{ord}(A \to B) := \max (\mathit{ord}(A) + 1, \mathit{ord}(B))$.
The \emph{arity} of a type $A = (A_1, \cdots, A_n, o)$, written $\arity(A)$, is defined to be $n$.

We assume an infinite set $\Var$ of typed variables, ranged over by $\Phi, \Psi, \phi, \psi, x, y, z$, etc.
{Raw} terms (ranged over by $M, N, P, Q,$ etc.) of the pure lambda calculus are defined by the grammar: $M ::= x \; | \; \lterm{x^A}{M} \; | \; (M\, N)$; by convention, applications associate to the left. 
Typing judgements (or \emph{terms-in-context}) have the form, $\seq{\Gamma}{M: A}$, where the environment $\Gamma$ is a list of variable bindings of the form $x : A$.
Henceforth, by a \emph{term} $M$ (respectively, $M : A$) we mean a well-typed term (respectively, of type $A$), i.e., $\seq{\Gamma}{M : A}$ is provable for some environment $\Gamma$ and type $A$.
We write $\FV(M)$ for the set of variables that occur free in $M$.


\subsection{Arenas, direct arena morphisms and justified sequences}
\label{sec:games-basics}

\begin{definition}[Arena]
An \emph{arena} is a triple
\(
A = \roundbra{\moves{A}, {\vdash_A}, \lambda_A}
\) 
such that $\moves{A}$ is a set (of moves), ${\vdash_A} \subseteq (\moves{A} +
  \makeset{\star}) \times \moves{A}$ is the enabling relation, and
$\lambda_A : \moves{A} \longrightarrow \makeset{\OO, \PP}$ is the ownership function that partitions moves into $\OO$-\emph{moves} and $\PP$-\emph{moves}, satisfying:
\begin{itemize}
\item For every $m \in \moves{A}$, there exists a unique $m' \in (\moves{A} \cup \makeset{\star})$ such that $m' \vdash_A m$; we call $m'$ the \emph{enabler} of $m$, or $m'$ \emph{enables} $m$.
Writing $\widehat{\vdash_A}$ for the inclusion of $\vdash_A$ in $(\makeset{\star} \cup \moves{A})^2$, we say that $m$ is \emph{hereditarily enabled} by $m'$ if $m' \mathbin{\widehat{\vdash_A}^\ast} m$, writing $\rho^\ast$ for the reflexive, transitive closure of a binary relation $\rho$). 

\item Whenever $m \vdash_A m'$ then $\lambda_A(m) \not= \lambda_A(m')$.

\end{itemize}
\end{definition}
We call a move \emph{initial} if its enabler is $\star$, and write $\Initial_A$ for the set of initial moves of $A$.
An arena $A$ is said to be $\OO$-\emph{initial} if $\lambda_A(m) = \OO$ for every initial $m$; 
and \emph{pointed} if $\Initial_A$ is a singleton set.
We define the set of $\OO$-\emph{moves} as \( \Omoves{A} :=\{ m \in \moves{A} \mid \lambda_A(m) = \OO \} \), and the set of $\PP$-\emph{moves} as \( \Pmoves{A} := \{ m \in \moves{A} \mid \lambda(m) = \PP \} \).  
The \emph{opposite arena} of $A$ is $A^\bot := (\moves{A}, \vdash_A, \lambda_A^\bot)$ where $\lambda^\bot_A(m) = \PP$ if, and only if, $\lambda_A(m) = \OO$.
\changed[lo]{We say that $A$ is \emph{ordered} if for each $m \in \moves{A}$, the set $\makeset{m' \mid m \vdash_A m'}$ is a linear order; 
and if so, we write $m \vdash_A^i m'$, where $i \geq 1$, to mean $m'$ is the $i$-th child of $m$.}

\bigskip

There is an obvious one-one correspondence between types and finite, ordered trees. 
For example, the type $(((o, o), o), o, (o, o, o), o)$ corresponds to the tree whose maximal (with respect to $\leq$) elements are $1 \sdot 1 \sdot 1$, $2$, $3 \sdot 1$ and $3 \sdot 2$.
We write $\Tree{A}$ to mean the tree that corresponds to the type $A$. 

\begin{definition}[Arena determined by type $A$]
Let $A$ be a type. The \emph{arena determined by $A$}, written $\Ar(A)$, is defined as follows:
\begin{itemize}
\item $\moves{\Ar(A)} := \Tree{A}$

\item $\star \vdash_{\Ar(A)} \epsilon$, and for all $\alpha, \beta \in \moves{\Ar(A)}$, $\alpha \vdash_{\Ar(A)} \beta \iff (\alpha \leq \beta \, \wedge \, |\beta| = |\alpha| + 1)$

\item $\lambda_{\Ar(A)}(\alpha) = \OO \iff |\alpha|$ even.
\end{itemize}
\end{definition}
Plainly $\Ar(A)$ is a finite arena which is ordered, pointed and $\OO$-initial.

We fix a scheme for naming nodes of a given tree (and hence moves of an arena determined by a type) using symbols of the following infinite ranked alphabet
\[
\Lam := \Var \cup
\makeset{\blambda^\alpha \mid \alpha \in \natnum^\ast} \cup
\makeset{\blambda x_1^{A_1} \cdots x_{n+1}^{A_{n+1}} \mid x_i^{A_i} \in \Var, n \geq 0}
\]
such that $\arity(x^A) := \arity(A)$, $\arity(\blambda x_1^{A_1} \cdots x_{n+1}^{A_{n+1}}) := 1$ and $\arity(\blambda^\alpha) := 0$. 
By abuse of language, we say that the \emph{lambda} $\blambda \xi_1^{A_1} \cdots \xi_n^{A_n}$ has type $(A_1, \cdots, A_n, o)$, and the \emph{dummy lambda} $\blambda^\alpha$ has type $o$.

\lo{N.B. (i) Dummy lambdas have arity 0 here (but not in $\LamAt$) because they only occur at leaves, but not so in $\LamAt$. (ii) Dummy lambdas are annotated with occurrences here. Hence $\Lam \not\subseteq \LamAt$.}

Given a type $A$ and an injective function, $\nu_A : \makeset{\alpha \in \Tree{A} \mid |\alpha| \hbox{ odd}} \to \Var$, such that 
whenever $\nu_A(\alpha) = x^B$ then $(\Tree{A})_{@\alpha} = \Tree{B}$, we extend $\nu_A$ to a function $\Tree{A} \to \Lam$ as follows. 
Let $\alpha \in \Tree{A}$ be of even length.
Suppose for each $\alpha \sdot i \in \Tree{A}$, we have $\nu_A(\alpha \sdot i) = x_i^{B_i}$, and $n = |\makeset{i \mid \alpha \sdot i \in \Tree{A}}|$, we define
\[
\nu_A(\alpha) :=
\left\{
\begin{array}{ll}
\blambda x_1^{B_1} \cdots x_n^{B_n} & \hbox{if $n > 0$}\\
\blambda^\alpha & \hbox{if $n = 0$}
\end{array} 
\right.
\]
It is straightforward to see that for every type $A$, the function $\nu_A : \Tree{A} \to \Lam$ satisfies the following:
\begin{inparaenum}
\item $\nu_A$ is injective,
\item $\nu_A$ defines a $\Lam$-labelled tree,
\item for all $\alpha \in \Tree{A}$, if $\nu_A(\alpha)$ has type $B$ then $(\Tree{A})_{@\alpha} = \Tree{B}$.
\end{inparaenum}
We call such a function $\nu_A$ a \emph{$\Lam$-representation} of $\Tree{A}$ (and for arena $\Ar(A)$, and type $A$).

\begin{example}\rm
Take the type $A \to A$ where $A = (((o, o), o), (o, o), o)$. We display the nodes of $\Tree{A \to A}$ via a naming scheme $\nu_{A \to A} : \Tree{A \to A} \to \Lam$ in Figure~\ref{fig:Lam-rep}. 
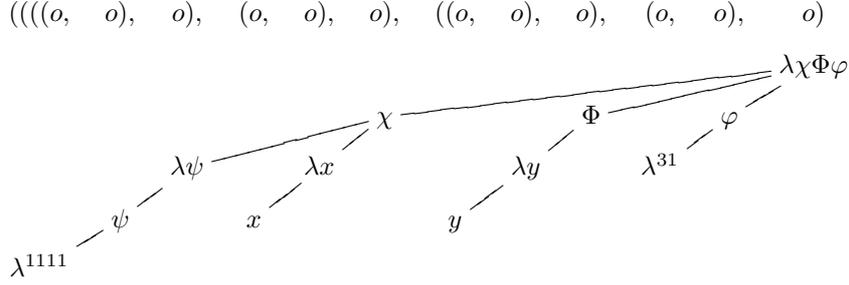
\begin{figure}[ht]
\[
\xymatrix@R-.7cm@C-.6cm{
((((o, & o), & o), & (o, & o), & o), & ((o, & o), & o), & (o, &  o),&  o) \\
 & & & & &   & & & & &   &  \blambda \chi \Phi \phi \\
 & & & & &   \chi \ar@{-}[rrrrrru]& & & \Phi \ar@{-}[rrru]& &  \phi \ar@{-}[ru]&  \\
 & & \blambda \psi \ar@{-}[rrru]& & \blambda x \ar@{-}[ru]& &  & \blambda y \ar@{-}[ru]& & \blambda^{31} \ar@{-}[ru]& &    \\
 & \psi \ar@{-}[ru] & & x \ar@{-}[ru] & & &  y \ar@{-}[ru]& & & & &    \\
\blambda^{1111} \ar@{-}[ru]& & & & &   & & & & &   &  \\
%
}
\]
\caption{A $\Lam$-representation of $\Tree{A \to A}$ where $A = (((o, o), o), (o, o), o)$. \label{fig:Lam-rep}}
\end{figure}
\end{example}

\subsubsection*{Product}
For arenas \( A \) and \( B \), we define the \emph{product arean} \( A \times
B \) by:
\begin{itemize}
\item \( \moves{A \times B} := \moves{A} + \moves{B} \),
\item \( \star \vdash_{A \times B} m \iff \star \vdash_{A} m \) or \( \star \vdash_{B} m \),
\item \( m \vdash_{A \times B} m' \iff m \vdash_{A} m' \) or \( m \vdash_{B} m' \),
\item \( \lambda_{A \times B}(m) := \left\{ 
\begin{array}{lr}
  \lambda_A(m) & \quad \hbox{if \( m \in \moves{A} \)}\\
  \lambda_B(m) & \quad \hbox{if \( m \in \moves{B} \)}
\end{array} 
\right. \)
\end{itemize}
Thus $A \times B$ is just the disjoint union of $A$ and $B$ \emph{qua} labelled directed graphs. 
For an indexed set \( \{ A_i \}_{i \in I} \) of arenas, their product \( \prod_{i \in I} A_i \) is defined similarly.

\subsubsection*{Function space}
For arenas \( A \) and \( B \), we define \emph{the function space arena} \( A \To B \) by:
\begin{itemize}
\item \( \moves{A \To B} := \moves{A} \times \Initial_B + \moves{B} \),
\item \( \star \vdash_{A \To B} m \iff \star \vdash_{B} m \),
\item \( m \vdash_{A \To B} m' \iff {} \)
  \begin{itemize}
  \item \( m \vdash_{B} m' \), or
  \item \( \star \vdash_{B} m \) and \( m' = (m'_A, m) \) and \( \star \vdash_{A} m_A' \), or
  \item \(m = (m_A, m_B) \) and \( m' = (m'_A, m_B) \) and \( m_A \vdash_{A} m_A' \)
  \end{itemize}
\item \( \lambda_{A \To B}(m) := 
\left\{ 
\begin{array}{ll}
  \lambda_A^\bot(m_A) & \quad \hbox{if \( m = (m_A, m_B) \in \moves{A} \times \Initial_B\)} \\
  \lambda_B(m) & \quad\hbox{if \( m \in \moves{B} \))}
\end{array} 
\right. \)
\end{itemize}
Observe that if $A$ and $B$ are types, then $\Ar(A \to B) = \Ar(A) \funsp \Ar(B)$.

\begin{definition}[Justified sequence]
A \emph{justified sequence of an arena \( A \)} is a finite sequence of moves, $m_1 \sdot m_2 \sdot \ldots \sdot m_n$, such that for each $j$, if $m_j$ is non-initial then $m_j$ has a pointer to $m_i$ such that $i < j$ and $m_i \vdash_A m_j$.
Formally it is a triple \( s = (\#s, s, \rho_s) \) consisting of a number \( \#s \in \natnum_0 \) (which is the length of the justified sequence), and total functions \( s : [\#s] \to \moves{A} \) (moves function) and \( \rho_s : [\#s] \to [\#s]_0 \) (pointers function) such that
\begin{itemize}
\item \( \rho_s(k) < k \) for every \( k \in [\#s] \), and
\item \( \rho_s \) respects the enabling relation: \( \rho_s(k) = 0 \) implies \( \star \vdash_{A} s(k) \), and 
\( \rho_s(k) \neq 0 \) implies \( s(\rho_s(k)) \vdash_{A} s(k) \).
\end{itemize}
As usual, by abuse of notation, we often write \( m_1 \sdot m_2 \dots m_n \) for a justified sequence such that \( s(i) = m_i \) for every \( i \), leaving the justification pointers implicit.  
Further we use \( m \) and \( m_i \) as meta-variables of \emph{move occurrences} in justified sequences.  
We write \( m_i \curvearrowleft m_j \) if \( \rho_s(j) = i > 0 \) and \( \star \curvearrowleft m_j \) if \( \rho_s(j) = 0 \).  
We call \( m_i \) the \emph{justifier of \( m_j \)}, and say $m_j$ is justified by $m_i$ whenever \( m_i \curvearrowleft m_j \).
We say that $m_j$ is \emph{hereditarily justified} by $m_i$ if $m_{i} \curvearrowleft^\ast m_{j}$. 
\changed[lo]{In case $A$ is an ordered arena, we write $m \stackrel{i}{\curvearrowleft} m'$ to mean $m \curvearrowleft m'$ and $m \vdash_A^i m'$.}
\end{definition}

It is convenient to relax the domain \( [\#s] = \{ 1, 2, \dots, \#s \} \) of justified sequences to arbitrary linearly-ordered finite sets such as a subset of \( [\#s] \).  
For example, given a justified sequence \( (\#s,\, s: [\#s] \to \moves{A},\, \rho_s : [\#s] \to [\#s]_0) \), consider a subset \( I \subseteq [\#s] \) that respects the justification pointers, 
i.e.,~\( k \in I \) implies \( \rho_s(k) \in I \cup \{ 0 \} \).  
Then the restriction \( (I, s{\upharpoonright_I} : I \to \moves{A}, \rho_s{\upharpoonright_I} : I \to \{ 0 \} \cup I) \) is a justified sequence in the relaxed sense.  
A justified sequence in the relaxed sense is identified with that in the strict sense through the unique monotone bijection \( \alpha : I \to [n] \).

A justified sequence is \emph{alternating} just if \( s(k) \in \Omoves{A} \iff k \) is odd. 
Henceforth we assume that justified sequences are alternating.

\begin{definition}[P-View / O-view]
Let \( m_1 \dots m_n \) be a justified sequence over an arena \( A \).  
Its \emph{P-view} \( \pview{m_1 \dots m_n} \) 
is a subsequence defined inductively by:
\begin{align*}\small
  \pview{m_1 \dots m_n} &:= \pview{m_1 \dots m_{n-1}}\,m_n  \hspace{28pt}\textrm{(if \( m_n \in \Pmoves{A} \))} \\
  \pview{m_1 \dots m_n} &:= m_n                           \hspace{72pt}\textrm{(if \( \star \curvearrowleft m_{n} \in \Omoves{A} \))} \\
  \pview{m_1 \dots m_n} &:= \pview{m_1 \dots m_k} \, m_n   \quad\textrm{(if \( m_k \curvearrowleft m_n \in \Omoves{A} \)).}
\end{align*}
Its \emph{O-view} \( \oview{m_1 \dots m_n} \) 
is a subsequence defined inductively by:
\begin{align*}\small
  \oview{m_1 \dots m_n} &:= \oview{m_1 \dots m_{n-1}}\,m_n  \hspace{28pt}\textrm{(if \( m_n \in \Omoves{A} \))} \\
  \oview{m_1 \dots m_n} &:= m_n                           \hspace{72pt}\textrm{(if \( \star \curvearrowleft m_{n} \in \Pmoves{A} \))} \\
  \oview{m_1 \dots m_n} &:= \oview{m_1 \dots m_k} \, m_n   \quad\textrm{(if \( m_k \curvearrowleft m_n \in \Pmoves{A} \)).}
\end{align*}
Formally the P-view of a justified sequence \( s \) is a subset \( I \subseteq [\#s] \).  
Then \( \pview{s} \) is the restriction of \( s \) to \( I \); similarly for $\oview{s}$.
Henceforth by a \emph{P-view}, we mean a justified sequence $s$ such that $\pview{s} = s$.

A P-move \( m_k \) in the sequence \( m_1 \ldots m_n \) (\( n \ge k \))
is \emph{P-visible} just if \( \star \curvearrowleft m_k \) or its justifier is in \( \pview{m_1 \dots m_k} \).
Similarly, an O-move \( m_k \) in the sequence \( m_1 \ldots m_n \) (\( n \ge k \))
is \emph{O-visible} just if \( \star \curvearrowleft m_k \) or its justifier is in \( \oview{m_1 \dots m_k} \).
A justified sequence \( s \) is \emph{P-visible} (respectively, \emph{O-visible}) just if each P-move (respectively, O-move) occurrence in \( s \) is P-visible (respectively, P-visible); $s$ is \emph{visible} just if it is both P- and O-visible.
If $s$ is a visible justified sequence, then so are $\pview{s}$ and $\oview{s}$ \citep{HO00}.
\end{definition}



\begin{definition}[Direct arena morphism]
A \emph{direct arena morphism}, $\calI : \roundbra{\moves{A}, \vdash_A, \lambda_A} \to \roundbra{\moves{B}, \vdash_B, \lambda_B}$, is a function $\calI : \moves{A} \to \moves{B}$ that respects: 
\begin{itemize}
\item enabling relation: for all $m, m' \in \moves{A}$, $\star \vdash_A m \iff \star \vdash_B \calI(m)$, and $m \vdash_A m' \iff \calI(m) \vdash_B \calI(m')$;
\item ownership: for all $m \in \moves{A}$, $\lambda_A(m) = \lambda_B(\calI(m))$.
\end{itemize}
\end{definition}

A direct arena morphism $\calI : A \to B$ induces a function on justified sequences by extension in the obvious way: 
take a justified sequence $t$ over $A$, define $
\calI^\ast(t) := (\#t, t', \rho_t)$ where $t'(i) := \calI(t(i))$ for $i \in [\#t]$.
It is straightforward to see that $\calI^\ast(t)$ is a justified sequence over $B$.
In fact, $t$ and $\calI^\ast(t)$ are \emph{isomorphic as justified sequences}, by which we mean that they are isomorphic directed graphs (viewing a justified sequence as a linear tree with justification pointers represented as back edges). I.e.~writing $\calI^\ast (t) = (\#{u}, u, \rho_{u})$, we have $\#t = \#{u}$, and $\rho_t = \rho_{u}$. 
(We do \emph{not}, however, require the map $\iota$ from the image of $t$ to the image of $u$, whereby $u(i) = \iota(t(i))$ for all $i \in [\#t]$, to be bijective.)
Thus it follows that $t$ is visible if, and only if, $\calI^\ast(t)$ is. 

\begin{definition}[Strong bijection induced by direct arena morphism]
Let $\calL$ and $\calL'$ be sets of justified sequences over arenas $A$ and $A'$ respectively.
Given a direct arena morphism $\calI : A \to A'$, we say that the map $\calI^\ast : \calL \to \calL'$ is a \emph{strong bijection (induced by $\calI$)} if $\calI^\ast$ restricted to $\calL$, $\calI^\ast \restriction \calL$, is injective, and the image of $\calI^\ast \restriction \calL$ is $\calL'$.
The adjective strong emphasises that, in addition to the bijectivity of $\calI^\ast$, for every $s \in \calL$, we have $s$ and $\calI^\ast(s)$ are isomorphic as justified sequences. 

\lo{N.B. It is possible for a (surjective but) non-bijective direct arena morphism to induce a bijective $\calI^\ast : \calL \to \calL'$ between sets of justified sequences. To show injectivity of $\calI^\ast$, we need to show: if $\calI^\ast(t) = p$ then for each $n$ such that $p \sdot n \in \calL'$, there is a unique $m$ such that $\calI^\ast(t) \sdot \calI(m) = p \sdot n$.}


Henceforth, whenever it is clear from the context, we write $\calI^\ast$ simply as $\calI$.
\end{definition}

\subsection{Game semantics of the lambda calculus}

A \emph{play} over an arena $A$ is a visible justified sequence.
A \emph{P-strategy} over arena $A$, or just \emph{strategy} for short, is a non-empty prefix-closed set $\sigma$ of plays over $A$ satisfying:
\begin{itemize}
\item \emph{Determinacy}. For every odd-length $s \in \sigma$, if $s \sdot a, s \sdot b \in \sigma$ then $a = b$.
\item For every even-length $s \in \sigma$, for every O-move $a$, if $s \sdot a$ is a play, then it is in $\sigma$.
\end{itemize}

We say that $\sigma$ is \emph{total} if for every odd-length $s \in \sigma$, there exists $a$ such that $s \sdot a \in \sigma$.
We say that $\sigma$ is \emph{innocent} if for every even-length $s \, a \in \sigma$ and for every odd-length $t \in \sigma$ if $\pview{s} = \pview{t}$ then $t \sdot a \in \sigma$.
It follows from the definition that an innocent strategy $\sigma$ is determined by the set of even-length P-views in $\sigma$, written $\pview{\sigma}$;
we say that $\sigma$ is \emph{compact} if $\pview{\sigma}$ is a finite set.


We can now organise arenas and innocent strategies into a category ${\mathbb I}$: 
objects are O-initial arenas; \lo{We don't seem to need O-initiality.}
and maps $\sigma : A \longrightarrow B$ are innocent strategies over arena $A \funsp B$.

\begin{theorem}
The category ${\mathbb I}$ is cartesian closed and enriched over CPOs.
\end{theorem}

Given a typing context $\Gamma = x_1 : A_1, \cdots, x_n : A_n$, we define $\mng{\Gamma} := ((\mng{A_1} \times \mng{A_2}) \cdots \mng{A_{n-1}}) \times \mng{A_n}$, with the empty context interpreted as the terminal object. 
The interpretation of a given term-in-context $\seq{\Gamma}{M : A}$ as a ${\mathbb I}$-map $\mng{\Gamma} \longrightarrow \mng{A}$ is standard, and we omit the definition. 



\begin{theorem}[Definability]
Let $A = (A_1, \cdots, A_n, o)$ be a type. For every compact, total, innocent strategy $\sigma$ over $\mng{A}$, there is a unique $\eta$-long $\beta$-normal form, $\seq{x_1 : A_1, \cdots, x_n : A_n}{M_\sigma : o}$, whose denotation is $\sigma$.
\end{theorem}

\begin{proof}
See \citep{HO00}.
\end{proof}

\section{Traversals over long forms}

\subsection{Long form of a lambda term}
The \emph{long form} of a term is a kind of canonical form, which is obtained by first constructing the $\eta$-long form of the term, 
and then replacing the standard binary application operators by full application operators $@^{\widehat{B}}$, which we called \emph{long application}.
The latter is achieved by replacing every subterm of the $\eta$-long form that has the form $(\lterm{x}{P})\, Q_1 \cdots Q_m : o$ by $@^{\widehat{B}} \,(\lterm{x}{P})\, Q_1 \cdots Q_m : o$ where $m \geq 1$. 
The type of the function symbol $@^{\widehat{B}}$ is $\widehat{B}$, which is a shorthand for $((B_1, \cdots, B_m, o), B_1, \cdots, B_m, o)$ where $B = (B_1, \cdots, B_m, o)$. 

\begin{definition}[Long form in concrete syntax]
Assume $\seq{\Gamma}{M : (A_1, \cdots, A_n, o)}$. The concrete syntax of the \emph{long form} of $M$, written $\lform{M}$, is defined by cases as follows.
\begin{enumerate}
\item $M$ is an application headed by a variable:
\[
\lform{x \, N_1 \cdots N_m} := \lterm{z_1^{A_1} \cdots z_n^{A_n}}{x \, \lform{N_1} \cdots \lform{N_m} \, \lform{z_1} \cdots \lform{z_n}}
\]
where $z_i \not\in \makeset{x} \cup \bigcup_{j=1}^m\FV\lform{N_j}$ for each $i$.
\item $M$ is an application headed by an abstraction with $m \geq 1$:
\[
\lform{ (\lterm{x}{P}) \, Q_1 \cdots Q_m } := 
\lterm{z_1^{A_1} \cdots z_n^{A_n}}{@^{\widehat{\Xi}} \, \lform{\lterm{x}{P}} \, \lform{Q_1} \cdots \lform{Q_m}   \, \lform{z_1} \cdots \lform{z_n}}
\]
where $\Xi = (B_1, \cdots, B_m, A_1, \cdots, A_n, o)$ which is the type of $\lterm{x}{P}$, and $z_i \not\in \FV\lform{\lterm{x}{P}} \cup \bigcup_{j=1}^m\FV\lform{Q_j}$ for each $i$.
\item $M$ is an abstraction:
\(
\lform{\lterm{x_1^{A_1} \cdots x_i^{A_i}}{P} } := 
\lterm{x_1^{A_1} \cdots x_i^{A_i}}{\lform{P}}
\).
\end{enumerate}
We shall elide the type superscript from variables $x^A$ and long application symbols $@^{\widehat{\Xi}}$, whenever it is clear form the context. 
We assume that bound variables in $\lform{M}$ are renamed afresh where necessary, so that if $\lterm{\overline x}{P}$ and $\lterm{\overline y}{Q}$ are distinct subterms (i.e.~they have different occurrences) then $\makeset{x_1, \cdots, x_m}$ and $\makeset{y_1, \cdots, y_n}$ are disjoint.
It is easy to verify that $\seq{\Gamma}{\lform{M} : A}$.
\end{definition}

If $M$ is $\beta$-normal, and so $\lform{M}$ has no occurrences of $@$, then $\lform{M}$ is essentially the $\eta$-long $\beta$-normal form of $M$.
Note that in $\lform{M}$ we additionally $\eta$-expand every ground-type subterm $P$ of $M$ to $\lterm{}{P}$ (we call $\blambda$ a ``dummy lambda'') provided $P$ occurs at an \emph{operand} position (meaning that $L \, P$ is a subterm of $M$ for some $L$). 
By a \emph{long form}, we mean the long form of a term.

\begin{example}\label{eg:longform1} \rm Consider the term-in-context $\seq{\phi : (((o, o), o, o), o, o, o), a : o}{M : (o, o)}$ where $M = \phi \, (\lterm{x^{(o, o)}}{x}) \, ((\lterm{y^o}{y}) \, a)$. We have
\[\lform{M} = \lterm{z_1}{\phi \, 
\big(
\lterm{x^{(o, o)} z^o}{x \, (\lterm{}{z})}
\big) 
\, 
\big(
\lterm{}{@ \, (\lterm{y^o}{y}) \, (\lterm{}{a})}
\big) 
\, 
\big(
\lterm{}{ z_1 }
\big)
} \; : \; (o, o).
\] 
\end{example}

We organise the abstract syntax tree (AST) of a long form, somewhat non-standardly, as a $\LamAt$-labelled (binding) tree where 
\[
\LamAt \; := \; \underbrace{\makeset{\blambda{\overline x} \mid \overline x = x_1 \cdots x_n \in \Var^\star}}_{\hbox{\emph{lambdas}}} \; \cup \; \underbrace{\Var \cup \makeset{@^{\widehat{A}} \mid A \in \mathit{Types}, \arity(A) > 0}}_{\hbox{\emph{non-lambdas}}}
\]
is a ranked alphabet such that $\arity(x^A) := \arity(A)$; 
$\arity(\blambda \overline x) = 1$,
and $\arity(@^{\widehat{A}}) := \arity(A) + 1$.
By construction, in the AST of a long form, nodes on levels 0, 2, 4, etc., are labelled by 
\emph{lambdas}, 
and nodes on levels 1, 3, 5, etc., are labelled by \emph{non-lambdas}.

\begin{definition}[$\Sigma$-labelled binding tree]\label{def:bindingtree}
\begin{enumerate}
\item A \emph{$\lambda$-alphabet} is a ranked alphabet $\Sigma$ which is  partitioned into $\Sigma_\lambda, \Sigma_\Var$ and $\Sigma_{\mathit{aux}}$, such that $\Sigma_\lambda$ consists of binders which have arity 1, $\Sigma_\Var$ consists of variables (whether bound or not), and $\Sigma_{\mathit{aux}}$ consists of the remaining auxiliary symbols.
\item A \emph{$\Sigma$-labelled binding tree} is a triple $(T, B, \ell)$ where $\Sigma$ is a $\lambda$-alphabet, $\ell : T \to \Sigma$ is a $\Sigma$-labelled tree, and $B : T \pto T$ is a partial function called \emph{binder}, satisfying: for all $\beta \in T$
\begin{description}
\item[(Bind)] $\ell(\beta) \in \Sigma_\Var \iff \beta \in \dom(B)$; and if $\beta \in \dom(B)$ then $B(\beta) < \beta$ and $\ell(\beta) \in \Sigma_\lambda$. 
By convention, if $\ell(\beta)$ is a free variable then $B(\beta) = \epsilon$.
\item[(Label)] The labelling function $\ell$ maps elements in $T$ of even lengths (including 0) into $\Sigma_\lambda$, and elements in $T$ of odd lengths into $\Sigma_\Var \cup \Sigma_\mathit{aux}$.
\end{description}
\end{enumerate}
\end{definition}

\begin{remark}
Our definition of binding tree 
is slightly more permissive than the original \citep{Stirling09}:
unlike Stirling, 
we do not assume that terms are closed.
\end{remark}


Observe that $\LamAt$ is a $\lambda$-alphabet: $\LamAt_\lambda$ consists of the lambdas, $\LamAt_\Var$ consists of the variables, and $\LamAt_\mathit{aux}$ consists of the long application symbols. 
The AST of a long form is a $\LamAt$-labelled binding tree.

\begin{example}\label{eg:longform2} \rm 
Take $M = \phi \, (\lterm{x^{(o, o)}}{x}) \, ((\lterm{y^o}{y}) \, a)$ of Example~\ref{eg:longform1}.
The abstract syntax tree of the long form 
\[\lform{M} = \lterm{z_1}{\phi \, 
\big(
\lterm{x^{(o, o)} z^o}{x \, (\lterm{}{z})}
\big) 
\, 
\big(
\lterm{}{@ \, (\lterm{y^o}{y}) \, (\lterm{}{a})}
\big) 
\, 
\big(
\lterm{}{ z_1 }
\big)
} \; : \; (o, o).
\]
is displayed in Figure~\ref{fig:lform-eg1}.
Let $(T, B, \ell)$ be the $\LamAt$-labelled binding tree representation of $\lform{M}$.
The binder function $B$ is indicated by the dotted arrows in the figure i.e.~\(\xymatrix@C=.7cm{m & m' \ar@{.>}@/_.5pc/[l]}\) means $B(m') = m$.
By convention (nodes that are labelled with) free variables are mapped by $B$ to the root node: thus $B: 1 \mapsto \epsilon, 12121 \mapsto \epsilon$.
\end{example}

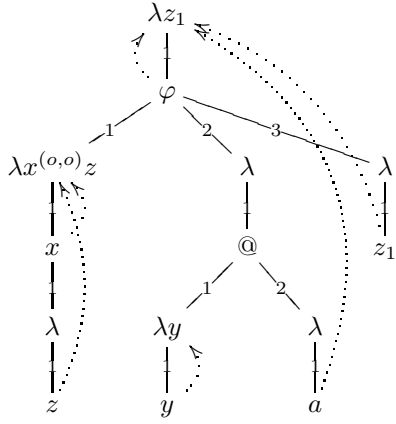
\begin{figure}[t]
\[\xymatrix@R-.3cm@C-.4cm{
 & \blambda z_1 \ar@{-}[d]|1  & & & \\ 
 & \phi \ar@{.>}@/^1pc/[u] \ar@{-}[ld]|1 \ar@{-}[rd]|2 \ar@{-}[rrrd]|3& & & \\
\blambda x^{(o, o)}z \ar@{-}[d]|1 
& & \blambda \ar@{-}[d]|1 
& & \blambda \ar@{-}[d]|1 
\\
x \ar@{-}[d]|1 \ar@{.>}@/_1pc/[u]& & @ \ar@{-}[ld]|1 \ar@{-}[rd]|2 & & z_1 \ar@{.>}@/_1.5pc/[llluuu]\\
\blambda \ar@{-}[d]|1 
& \blambda y \ar@{-}[d]|1 
& & \blambda \ar@{-}[d]|1 
& \\
z \ar@{.>}@/_1pc/[uuu]& y \ar@{.>}@/_1pc/[u]& & a \ar@{.>}@/_3pc/[lluuuuu] &}
\]
\caption{The long form of the term $M$ in Example~\ref{eg:longform1}\label{fig:lform-eg1}}
\end{figure}

\begin{lemma}\label{lem:bindingtree-char}
A finite $\LamAt$-labelled binding tree $(T, B, \ell)$ is the long form of a term if, and only if, it satisfies the following labelling axioms:
\begin{description}

\item[(Lam)] If $\ell(\beta) = y$ and $\ell(B(\beta)) = \blambda \overline x$ then $y \in \overline x$ (i.e.~$y$ is bound in the term) or $B(\beta) = \epsilon$ ($y$ is a free variable).

\item[(Leaf)] A node $\alpha$ is maximal in $T$ if, and only if, $\ell(\alpha)$ is a ground-type variable.

\item[(TVar)] A node labelled by $@^{\widehat{A}}$ where $A = (A_1, \cdots, A_n, o)$ has $n+1$ children with lambda labels of types $A, A_1, \cdots, A_n$ respectively (see left of Figure~\ref{fig:TVarLA}).

\item[(T@)] A node labelled by a {variable} $\phi : (A_1, \cdots, A_n, o)$ has $n$ children with lambda labels of types $A_1, \cdots, A_n$ respectively  (see right of Figure~\ref{fig:TVarLA})
\end{description}
\end{lemma}

\begin{figure}[ht]
\centering 
\(
\xymatrix@R-.3cm@C-.2cm{ & & @^A
 \ar@{-}[dll]|1 \ar@{-}[dl]|2 \ar@{-}[dr]|{n+1} & & \quad \qquad & \phi \ar@{-}[dl]|1 \ar@{-}[dr]|n & \\ 
 \blambda \xi_1^{A_1} \cdots
 \xi_n^{A_n} & \blambda \overline{\eta_1} : A_1 & 
 \cdots & \blambda \overline{\eta_n} : A_n
 & \blambda
 \overline{\eta_1} : A_1 & \cdots & \blambda \overline{\eta_n} : A_n
 }
\)
\caption{Labelling rules for long application symbols and variables. \label{fig:TVarLA}}
\end{figure}
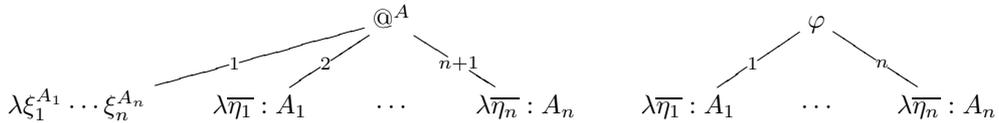

\begin{proof} 
The direction ``$\Rightarrow$'' can be proved straightforwardly by  induction on the rules that define $\lform{M}$.
For ``$\Leftarrow$'', take a $\LamAt$-labelled binding tree $(T, B, \ell)$ that satisfies the axioms. Because of (Label), $\ell(\epsilon)$ has the form $\blambda x_1 \cdots x_m$; and because of (Leaf), every maximal element in $T$ has an odd length.
The base cases are therefore long forms of the shape $\lterm{\overline x}{y}$ with $B : 1 \mapsto \epsilon$.
For the inductive cases, $\ell(1)$ is either $@^{\widehat{A}}$ or $x^A$ where $A = (A_1, \cdots, A_n, o)$ with $n \geq 1$.
Suppose the former.
By assumption, for all $\beta \in \dom(B)$, $B(\beta) < \beta$. 
Let $\alpha \in T$. The tree \emph{$T$ rooted at $\alpha$}, denoted $T_{@\alpha}$, is the set $\makeset{\gamma \mid \alpha \sdot \gamma \in T}$. 
Define $B_{@\alpha} : T_{@\alpha} \pto T_{@\alpha}$ by setting $\dom(B_{@\alpha}) := \makeset{\beta \mid \alpha \sdot \beta \in \dom(B)}$; and
$B_{@\alpha}(\beta) := \beta'$ if $B(\alpha \sdot \beta) = \alpha \sdot \beta'$, and 
$B_{@\alpha}(\beta) := \epsilon$ if $B(\alpha \sdot \beta) < \alpha$.
Next define $\ell_{@\alpha} : T_{@\alpha} \to \LamAt$ by $\gamma \mapsto \ell(\alpha \sdot \gamma)$.
It follows that for each $i \in [n+1]$, $(T_{@1 \sdot i}, B_{@1 \sdot i}, \ell_{@1 \sdot i})$ is a $\LamAt$-labelled binding tree that satisfies the axioms. 
By the induction hypothesis, suppose the binding trees are the ASTs of the long forms $\lform{M}, \lform{N_1}, \cdots, \lform{N_n}$ respectively.
Then $(T, B, \ell)$ is the AST of the long form $\lterm{\overline x}{@ \, \lform{M} \, \lform{N_1} \cdots \lform{N_n}}$ where $\ell(\epsilon) = \blambda \overline x$.
The latter case of $\ell(1) = x^A$ is similar.
\end{proof}

\subsection{Explicit arenas, succinct arenas and succinct long form}

Given a term, we identify two arenas that are associated with $\lform{M}$, namely, explicit arena $\ArExp\lform{M}$, and succinct arena $\ArSuc\lform{M}$.
The former arena provides the setting for traversals: 
we will define traversals over $\lform{M}$ as justified sequences over the explicit arena $\ArExp\lform{M}$ that satisfy certain constraints.

Fix a term-in-context $\seq{z_1 : A_1, \cdots, z_i : A_i }{M : (A_{i+1}, \cdots, A_n, o)}$, and write $A = (A_1, \cdots, A_n, o)$ and $\lform{M} = (T, B, \ell)$. 
Furthermore assume that $\makeset{(\gamma_1, \Xi_1), \cdots, (\gamma_r, \Xi_r)} = \makeset{(\gamma, \Xi) \mid \ell(\gamma) = @^{\widehat{\Xi}}}$.

\begin{definition}[Explicit Arena of $\lform{M}$, $\ArExp\lform{M}$]
The \emph{explicit arena} of $\lform{M} = (T, B, \ell)$, $\ArExp\lform{M}$, is defined as follows.
The underlying move-set $\moves{\ArExp\lform{M}} := T$.
The enabling relation, $\vdash_{\ArExp\lform{M}}$, is defined as follows: 
\begin{itemize}
\item $\star \vdash_{\ArExp\lform{M}} \epsilon$; 
and for all $i \in [r]$, $\star \vdash_{\ArExp\lform{M}} \gamma_i$.

\item If $\ell(\alpha) = x^{(B_1, \cdots, B_m, o)}$ then for each $i \in [m]$, $\alpha \vdash_{\ArExp\lform{M}} \alpha \sdot i$.

\item If $\ell(\alpha) = @^{\widehat{C}}$ where $C = (C_1, \cdots, C_l, o)$ then for each $i \in [l+1]$, $\alpha \vdash_{\ArExp\lform{M}} \alpha \sdot i$.

\item If $\ell(\alpha) = \blambda x_1 \cdots x_m$ then: for all $\alpha' \in B^{-1}(\alpha)$ 
\begin{itemize}
\item if $\ell(\alpha') = x_i$ for some $i \in [m]$ then $\alpha \vdash_{\ArExp\lform{M}} \alpha'$. 
\lo{In which case we write $\alpha \vdash_{\ArExp\lform{M}}^i \alpha'$, with the superscript $i$, to indicate $\ell(\alpha') = x_i$} 
\item if $\ell(\alpha') \not\in \makeset{x_1, \cdots, x_m}$ (it follows from (Lam) that $\alpha = \epsilon$ and $\ell(\alpha') = z_k \in \FV(M)$ for some $k \in [i]$) then $\epsilon \vdash_{\ArExp\lform{M}} \alpha'$. 
\lo{In which case we write $\epsilon \vdash_{\ArExp\lform{M}}^k \alpha'$, with the superscript $k$, to indicate $\ell(\alpha') = z_k$.}
\end{itemize} 
\end{itemize}
Finally $\lambda_{\ArExp\lform{M}}(\alpha) = \OO \iff |\alpha|$ is even. 
\end{definition}

Thus the explicit arena $\ArExp\lform{M}$ has the same underlying node-set $T$ as the long form $\lform{M} = (T, B, \ell)$.
Every lambda-labelled node is enabled by its predecessor in $T$.
A variable-labelled node $\alpha$ is enabled by its binder $B(\alpha)$; by convention, a node $\alpha$ labelled by a free variable is enabled by the root $\epsilon = B(\alpha)$.
A node is an O-move if and only if its label is a lambda.





The succinct arena of a long form is part of a compact representation of the long form as a binding tree, called succinct long form.


\begin{definition}[Succinct Arena of $\lform{M}$, $\ArSuc\lform{M}$]
The \emph{succinct arena} of $\lform{M}$, $\ArSuc\lform{M}$, is defined to be the arena $\Ar(A) \, \times \, \prod_{i=1}^r \Ar^\bot(\widehat{\Xi_i})$. 
\end{definition} 

The arena $\ArSuc\lform{A}$ is a disjoint union of $\Ar(A), \Ar^\bot(\widehat{\Xi_1}), \cdots, \Ar^\bot(\widehat{\Xi_r})$, \emph{qua} labelled directed graphs.
Notice that $\ArSuc\lform{M}$ depends only on the list of types 
(viz.~$A, \Xi_1, \cdots, \Xi_r$) that occur in $\lform{M}$, and not on the size of $M$; furthermore, in case $M$ is $\beta$-normal, $\ArSuc\lform{M} = \Ar(A)$.

Next we define a function $\hat{\ell} : |\ArExp\lform{M}| \to |\ArSuc\lform{M}|$ where 
\[
|\ArSuc\lform{M}| := \makeset{\epsilon} \times \moves{\Ar(A)}
\; \cup \;
\bigcup_{i=1}^r \makeset{\gamma_i} \times \moves{\Ar^\bot(\widehat{\Xi_i})}
\] 
is a (convenient) representation of the underlying set of the arena $\ArSuc\lform{M}$.
To aid the definition of $\hat{\ell}$, we use a predicate 
\(
S \; \subseteq \; T \times \LamAt \times |\ArSuc\lform{M}|
\). 
The idea is that $(\alpha, s, \gamma, \beta) \in S$ means: 
\begin{inparaenum}[(i)]
\item $\alpha \in T$ and $\ell(\alpha) = s$, and 
\item $\alpha$ is hereditarily enabled by $\gamma$ in the arena $\ArExp\lform{M}$ where $\gamma \in \makeset{\epsilon, \gamma_1, \cdots, \gamma_r}$, and
\item $\alpha$ is mapped by $\hat{\ell}$ to $(\gamma, \beta)$. 
\end{inparaenum}

The predicate $S$ is defined by induction over the following rules:
\begin{itemize}
\item $(\epsilon, \ell(\epsilon), 
\epsilon, \epsilon) \in S$; 
for all $i \in [r]$, $(\gamma_i, \ell(\gamma_i), \gamma_i, \epsilon) \in S$.

\item If $(\alpha, x^{(B_1, \cdots, B, o)}, \gamma, \beta) \in S$ then for all $i \in [m]$, $(\alpha \sdot i, \ell(\alpha \sdot i), \gamma, \beta \sdot i) \in S$. \lo{Notice that, by construction of long form, $\ell(\alpha \sdot i)$ has the form $\blambda y_1^{C_1} \cdots y_l^{C_l}$ where $B_i = (C_1, \cdots, C_l , o)$.}

\item If $(\gamma, @^{\widehat{C}}, \gamma, \epsilon) \in S$ where $C = (C_1, \cdots, C_l, o)$ then for all $i \in [l+1]$, $(\gamma \sdot i, \ell(\gamma \sdot i), \gamma, i) \in S$. 
\lo{Notice that $\ell(\gamma \sdot 1)$ has the form $\blambda y_1^{C_1} \cdots y_l^{C_l}$; 
and for all $i \in \makeset{2, \cdots, l+1}$, $\ell(\gamma \sdot i)$ has the form $\blambda z_1^{B_1} \cdots z_m^{B_m}$ where $C_{i-1} = (B_1, \cdots, B_m, o)$. }

\item If $(\alpha, \blambda x_1^{B_1} \cdots x_m^{B_m}, \gamma, \beta) \in S$ then: for all $\alpha' \in B^{-1}(\alpha)$
\begin{itemize}
\item if $\ell(\alpha') = x_i^{B_i}$ for some $i \in [m]$ then $(\alpha', \ell(\alpha'), \gamma, \beta \sdot i) \in S$
\item if $\ell(\alpha') \not\in \makeset{x_1, \cdots, x_m}$ then (by (Lam)) $\alpha = \epsilon$ and $\ell(\alpha') = z_k^{A_k} \in \FV(M)$ for some $k \in [i]$, and we have $(\alpha', \ell(\alpha'), \epsilon, k) \in S$

\end{itemize}
\end{itemize}
Then, by a straightforward induction over the length of $\alpha$, we have: for all $\alpha \in T$, there exist unique $\gamma$ and $\beta$ such that $(\alpha, \ell(\alpha), \gamma, \beta) \in S$.


\begin{definition}[Succinct Long Form]\label{def:clform}
We organise $|\ArSuc\lform{M}|$ into a $\lambda$-alphabet as follows:
\begin{align*}
|\ArSuc\lform{M}|_\lambda := {} & \makeset{(\epsilon, \alpha) \mid \alpha \in \moves{\Ar(A)}, \hbox{$|\alpha|$ even}} \; \cup \\
& \bigcup_{i=1}^r \makeset{(\gamma_i, \alpha) \mid \alpha \in |\Ar^\bot(\widehat{\Xi_i})|, \hbox{$|\alpha|$ odd}} \\
|\ArSuc\lform{M}|_\Var := {} & \makeset{(\epsilon, \alpha) \mid \alpha \in \moves{\Ar(A)}, \hbox{$|\alpha|$ odd}} \cup {}\\
& \bigcup_{i=1}^r \makeset{(\gamma_i, \alpha) \mid \alpha \in |\Ar^\bot(\widehat{\Xi_i})|, |\alpha| > 0, \hbox{$|\alpha|$ even}}\\
|\ArSuc\lform{M}|_\mathit{aux} := {} &
\makeset{(\gamma_1, \epsilon), \cdots, (\gamma_r, \epsilon)}
\end{align*}
The \emph{succinct long form} of a term $M$ is the $|\ArSuc\lform{M}|$-labelled binding tree, $(T, B, \hat{\ell})$, where 
\(
\hat{\ell} : T \to 
|\ArSuc\lform{M}|
\)
is given by $\hat{\ell}(\alpha) = (\gamma, \beta)$ just if $(\alpha, \ell(\alpha), \gamma, \beta) \in S$.

\end{definition}

\lo{N.B. (i) The even maximal nodes in $\moves{\Ar(A)}$ and the odd maximal nodes in $|\Ar^\bot(\widehat{\Xi_i})|$ are the dummy lambdas; the root nodes in $|\Ar^\bot(\widehat{\Xi_i})|$ are the long-apply nodes.

(ii) It is straightforward to check that $(T, B, \hat{\ell})$ is a $|\ArSuc\lform{M}|$-labelled binding tree: ${\hat{\ell}(\beta) \in |\ArSuc\lform{M}|_\Var}$ iff ${\beta \in \dom(B)}$; and if $\beta \in \dom(B)$ then $\ell(\beta) \in |\ArSuc\lform{M}|_\lambda$.}

\begin{lemma}
For every term $M$, the function $\hat{\ell} : \moves{\ArExp\lform{M}} \to \moves{\ArSuc\lform{M}}$ gives a direct arena morphism from $\ArExp{\lform{M}}$ to $\ArSuc\lform{M}$. 
In general $\hat{\ell}$ is neither injective nor surjective. 
\end{lemma}

\begin{proof}
It follows from the definition that $\star \vdash_{\ArExp\lform{M}} \alpha \, \iff \, \star \vdash_{\ArSuc\lform{M}} \hat{\ell}(\alpha)$.

To verify $\alpha \vdash_{\ArExp\lform{M}} \alpha' \iff \hat{\ell}(\alpha) \vdash_{\ArSuc\lform{M}} \hat{\ell}(\alpha')$, we analyse the cases of $\ell(\alpha)$.
To illustrate, consider the case of $\ell(\alpha) = \blambda x_1 \cdots x_m$. 
Take $\alpha' \in B^{-1}(\alpha)$. 
If $\ell(\alpha') = x_i$ for some $i \in [m]$ then $\alpha \vdash_{\ArExp\lform{M}} \alpha'$.
Suppose $\hat{\ell}(\alpha) = (\gamma, \beta)$, say. 
Then $\hat{\ell}(\alpha') = (\gamma, \beta \sdot i)$, and we have $\hat{\ell}(\alpha) \vdash_{\ArSuc\lform{M}} \hat{\ell}(\alpha')$ as desired.
If for some $k \in [i]$, $\ell(\alpha') = z_k \not\in \makeset{x_1, \cdots, x_m}$, then $\alpha = \epsilon$, and $\alpha \vdash_{\ArExp\lform{M}} \alpha'$. Now $\hat{\ell}(\alpha) = (\epsilon, \epsilon)$ and $\hat{\ell}(\alpha') = (\epsilon, k)$, and we have $\hat{\ell}(\alpha) \vdash_{\ArSuc\lform{M}} \hat{\ell}(\alpha')$ as desired.

Finally, to show $\lambda_{\ArSuc\lform{M}}(\hat{\ell}(\alpha)) = \lambda_{\ArExp\lform{M}}(\alpha)$, notice that $\lambda_{\ArExp\lform{M}} (\alpha) = \OO \iff |\alpha|$ even.
Let $\hat{\ell}(\alpha) = (\gamma, \beta)$ and suppose $|\alpha|$ is even.
It is straightforward to see that if $\gamma = \epsilon$ then $|\alpha| \equiv |\beta| \mod 2$, and so $\lambda_{\Ar(A)}(\beta) = \OO = \lambda_{\ArSuc\lform{M}} (\hat{\ell}(\alpha))$. And
if $\gamma = \gamma_i$ for some $i \in [r]$ then $|\alpha| \equiv |\beta| + 1  \mod 2$,
and so $\lambda_{\Ar^\bot(\widehat{\Xi_i})}(\beta) = \OO = \lambda_{\ArSuc\lform{M}} (\hat{\ell}(\alpha))$.
The other case (i.e.~$|\alpha|$ is odd) is symmetric.
\end{proof}

Take a term-in-context $\seq{}{M : A}$. The direct arena morphism $\hat{\ell} : \ArExp\lform{M} \to \ArSuc\lform{M}$ is closely related to the game semantics of $M$, $\mng{\seq{}{M}}$, an innocent strategy over the arena $\mng{A}$. The function $\hat\ell$ maps nodes of the tree $|\ArExp\lform{M}|$ to moves of the arena $\mng{A}$.
We illustrate this map in the following example.

\begin{example}
Take $\seq{}{M : A}$ where $M = \lterm{x}{x \, (\lterm{y}{y}) \, (\lterm{z}{x \, (\lterm{y'}{y'}) \, (\lterm{z'}{z'})})}$ and $A = (((o, o), (o, o), o), o)$.
Since $M$ is $\beta$-normal, we have $\ArSuc\lform{M} = \Ar(A)$.
In the following, we display the direct arena morphism $\hat\ell : |\ArExp\lform{M}| \to |\ArSuc\lform{M}|$ by annotating $\hat{\ell}(\alpha)$ next to the node $\alpha$, separated by $:$.
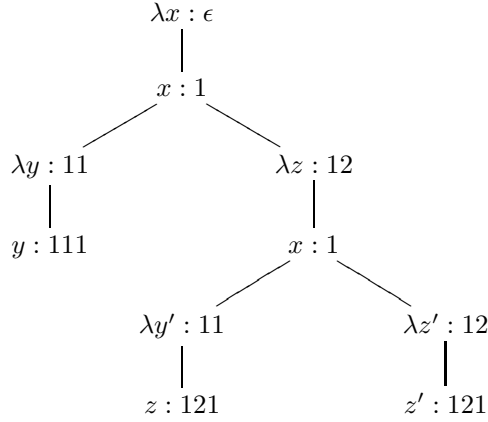
\[\xymatrix@R-.3cm@C-.4cm{
 & {\blambda x} : {\epsilon} \ar@{-}[d] & & \\
 & x:{1} \ar@{-}[dl] \ar@{-}[dr] & & \\ 
{\blambda y}:{11} \ar@{-}[d] & & \blambda z:{12} \ar@{-}[d] & \\
 y:{111} & & x:{1} \ar@{-}[dl] \ar@{-}[dr] & \\
& {\blambda y'}:{11}  \ar@{-}[d] & & {\blambda z'}:{12}  \ar@{-}[d] \\
 & z:{121} & & {z'}:{121}}
\]

\end{example}

\begin{remark}[Every path in the long form is a P-view]\label{rem:pview-path}
Take a long form $\lform{M} = (T, B, \ell)$. 
Every path $\alpha_1 \sdot \alpha_2\sdot \cdots \sdot \alpha_n$ in the tree $T$ is the underlying sequence of a (unique) P-view $p$ over the arena $\ArExp\lform{M}$, whose pointers are defined as follows.
Suppose $i \geq 3$ is odd; then the O-move $\alpha_i$ (where $\ell(\alpha_i)$ is necessarily a lambda) is justified by $\alpha_{i-1}$; note that we have $\alpha_{i-1} \vdash_{\ArExp\lform{M}} \alpha_i$. 
Suppose $i$ is even; if $\ell(\alpha_i)$ is a variable, then the P-move $\alpha_i$ is justified by $B(\alpha_i)$ -- note that $B(\alpha_i) < \alpha_i$; otherwise, $\ell(\alpha_i) = @$, and $\alpha_i$ is initial in $\ArExp\lform{M}$. 
Since $\hat{\ell}$ is a direct arena morphism, $\hat{\ell}(p)$ is a P-view over the arena $\ArSuc\lform{M}$.
\end{remark}




\begin{example}\label{eg:kierstead}
Consider $\lform{K} = \lterm{f}{f(\lterm{x}{f(\lterm{y}{f(\lterm{z}{\lterm{}{x}})})})} : (((o, o), o), o)$. 
\emph{Qua} $\LamAt$-labelled binding tree $(T, B, \ell)$, the tree $T$ consists of all prefixes of $1^7 = 1 \sdot 1 \sdot 1 \sdot 1 \sdot 1 \sdot 1 \sdot 1$ 
and $B : 1 \mapsto \epsilon, 1^3 \mapsto \epsilon, 1^5 \mapsto \epsilon, 1^7 \mapsto 1^2$.
Take the P-view $p$ whose underlying sequence of moves is the maximal path of the tree $T$. 
(In the following, pointers from O-moves are not displayed.)

\noindent The P-view $p$ over the explicit arena $\ArExp\lform{K}$ is
\begin{equation}
p = \Pstr[0.7cm]{(a){\epsilon} \; \sdot \; (b-a){1} \; \sdot \; (c){1^2} \; \sdot \; (d-a){1^3} \; \sdot \; 1^4 \; \sdot \; (e-a){1^5} \; \sdot \; 1^6 \; \sdot \; (f-c){1^7}}
\label{eq:kierstead-dom}
\end{equation}
The ``justified sequence'' of the $\LamAt$-labels traced out by $p$, $\ell^\ast(p)$,  is
\begin{equation}
\ell^\ast(p) = \Pstr[0.7cm]{(a){\blambda f} \; \sdot \; (b-a){f} \; \sdot \; (c){\blambda x} \; \sdot \; (d-a){f} \; \sdot \; \blambda y \; \sdot \; (e-a){f} \; \sdot \; \blambda z \; \sdot \; (f-c){x}}
\label{eq:kierstead-im}
\end{equation}
The P-view, ${\hat{\ell}}^\ast(p)$, over the succinct arena $\ArSuc\lform{M} = \Ar \, (((o, o), o), o)$ is
\begin{equation}
\hat{\ell}^\ast(p) = \Pstr[0.7cm]{(a){\epsilon} \; \sdot \; (b-a){1} \; \sdot \; (c){1^2} \; \sdot \; (d-a){1} \; \sdot \; 1^2 
\; \sdot \; (e-a){1} \; \sdot \; 1^2 \; \sdot \; (f-c){1^3}}
\label{eq:kierstead-ar}
\end{equation}
The P-view, $\hat{\ell}^\ast(p)$, over a $\Lam$-representation of 
$\ArSuc\lform{M}$ is
\begin{equation}
\hat{\ell}^\ast(p) = \Pstr[0.7cm]{(a){\blambda f} \; \sdot \; (b-a){f} \; \sdot \; (c){\blambda x} \; \sdot \; (d-a){f} \; \sdot \; \blambda x 
\; \sdot \; (e-a){f} \; \sdot \; \blambda x \; \sdot \; (f-c){x}}
\label{eq:kierstead-ar}
\end{equation}
Note that the set of elements that occur in $p$ (respectively $\ell^\ast(p)$ and ${\hat{\ell}}^\ast(p)$) has size 8 (respectively 6 and 4).

\end{example}

\subsection{Traversals over a long form}
\emph{Henceforth we write a long form \emph{qua} $\LamAt$-labelled binding tree as $\lform{M} = (T_M, B_M, \ell_M)$.} 

\begin{definition}[Traversals] \rm\label{def:traversal}\emph{Traversals} over a long form $\lform{M}$ are justified sequences over the arena $\ArExp \lform{M}$ defined by induction over the rules in Table~\ref{tab:traversals}. 
We write $\trmng{M}$ for the set of traversals over $
\lform{M}$.
It is convenient to refer to elements $\alpha \in |\ArExp\lform{M}| = T_M$ by their labels $\ell_M(\alpha)$.
We shall do so in the following whenever what we mean is clear from the context. 
\end{definition}

\begin{remark}\label{rem:traversals}
\begin{asparaenum}[(i)]
\item The rule (App) says that if a traversal ends in a @-labelled node $n$, 
then the traversal extended with the first (left-most) child of $n$ is a traversal.

\item The rule (Lam) says that if a traversal $t$ ends in a $\blambda$-labelled node $n$, then $t$ extended with the child node of $n$, $n'$, is also a traversal. 
To illustrate how the pointer of $n'$ in $t \sdot n'$ is determined, 
first note that every path in the (abstract syntax tree of the) long form $\lform{M}$ is a justified sequence which is a P-view (Remark~\ref{rem:pview-path}). 
Take, for example, traversal (\ref{eq:traversal-intro-1}) truncated at the 15-th move---call it $t \sdot \lambda$.  
By (Lam), $t \sdot \blambda \sdot y$ is a traversal, where $y$ is the child node of $\blambda$.
Notice that there are two occurrences of $\lambda f y$ (5th and 9th move respectively) in $t$ to which $y$ could potentially point.
However, in $\pview{t \sdot \lambda \sdot y} = 
\Pstr[0.65cm]{
(n1){\blambda}\, 
(n2){@}\, 
(n3-n2,30:2){\blambda f y}\, 
(n4-n3){f}\, 
(n5-n4,30:1){\blambda}\, 
(n6-n1){g}\,
(n7-n6,30:2){\blambda}\, 
(n8-n3){y}\,
}
$, which is a path in the long form $\lform{N \, P \, R}$, $y$ is bound by the 3rd move.
Since it is the 5th move of $t \sdot \lambda \sdot y$ which is mapped by $\pview{\hbox{-}}$ to the 3rd move of $\pview{t \sdot \lambda \sdot y}$, (Lam) says that $y$ in $t \sdot \lambda \sdot y$ points to the 5th-move.

\item If a traversal ends in a node labelled with a variable $\xi_i$,
then there are two cases, corresponding to whether $\xi_i$ is hereditarily justified by a bound (BVar) or free (FVar) variable in the long form $\lform{M}$.
Observe that some pointers in (BVar) 
are labelled; 
for example, 
\(
\Pstr[0.65cm]{
(n1){n}\; 
\cdots\; 
(n2-n1,30:i+1){\blambda \overline \eta}
}
\)
means that the node $\blambda \overline \eta$ is the $(i+1)$-th child of $n$. 
Intuitively, the rules (BVar) capture the switching of control between caller and callee, or between formal and actual parameters.
Thus, in rule (BVar).2, $\xi_i$ is the $i$-th formal parameter, and $\blambda \overline \eta$---the $i$-th child of $n$---the (root of the) $i$-th actual parameter. 
Note, however, that in rule (Bvar).1, although $\blambda \overline \eta$ is the $(i+1)$-th child, it is actually the $i$-th actual parameter because the 1st-child of $@$ is not the 1st actual parameter, but rather the body of the function call itself.

\item The rule (FVar) is the only rule that permits traversals to branch, and grow in different directions.
\end{asparaenum}
\end{remark}

\begin{table}[ht]
\begin{center}
\mbox{\begin{shadowbox}[13.5cm]
\begin{itemize}
\item[(\emph{Root})]  
$\epsilon \in \trmng{M}$.

\item[(\emph{App})] 
If $\; t \sdot @ \; \in \; \trmng{M}$ then
$\Pstr[0.7cm]{t \sdot (app){@} \sdot (lx-app,45:1){\blambda \overline{\xi}}} \; \in \; \trmng{M}$.

\item[(\emph{BVar})] 
If $\Pstr[0.7cm]{t \sdot (n){n} \sdot (lxi){\blambda \overline\xi} \cdots (xi-lxi,45){\xi_i} } \; \in \; \trmng{M}$ where $\overline \xi = \xi_1 \cdots \xi_n$ and $\xi_i$ is hereditarily justified by an $@$ then
\begin{enumerate} 
\item if $n$ is (labelled by) $@$ then $\Pstr[0.7cm]{t \sdot (n){n} \sdot (lxi){\blambda \overline\xi} \cdots (xi-lxi,45){\xi_i} \sdot (leta-n,45:{i+1}){\blambda\overline\eta}}  \; \in \; \trmng{M}$
\item if $n$ is (labelled by) a variable then $\Pstr[0.7cm]{t \sdot (n){n} \sdot (lxi){\blambda \overline\xi} \cdots (xi-lxi,45){\xi_i} \sdot (leta-n,45:i){\blambda\overline\eta}}  \; \in \; \trmng{M}$.
\end{enumerate}

\item[(\emph{FVar})] 
If $\;t \sdot \xi  \; \in \trmng{M}$ and the variable $\xi $ is \emph{not} hereditarily justified by an $@$ (equivalently, $\xi$ is hereditarily justified by the opening node $\epsilon$) then
$\Pstr[0.7cm]{t \sdot (d){\xi} \sdot (leta-d,45:j){\blambda \overline\eta}} \; \in \trmng{M}$, for every $1 \leq j \leq \arity(\xi)$. 

\item[(\emph{Lam})] 
If $\;t \sdot \blambda\overline{\xi} \; \in \; \trmng{M}$ and let $n$ be the (unique) child node of $\blambda \overline{\xi}$ in $T_M$, then 
$\;t\sdot \blambda\overline{\xi} \sdot n  \; \in\;  \trmng{M}$.
By a straightforward induction, $\pview{t \sdot \blambda \overline \xi \sdot n}$ is a (justified) path in the tree $T_M$. If $n$ is labelled by a variable (as opposed to $@$) then its pointer in $t \sdot \blambda \overline \xi \sdot n$ is determined by the justified sequence $\pview{t \sdot \blambda \overline \xi \sdot n}$ which is guaranteed to be a path in $T_M$.
\changed[lo]{Precisely if in the P-view $\pview{t \sdot \blambda \overline \xi \sdot n}$, $n$ points to the $i$-th move, then $n$ in $t \sdot \blambda \overline \xi \sdot n$ points to the $j$-th move, where the $j$-th move is the necessarily unique move-occurrence that is mapped to the $i$-th move under the P-view transformation: $\pview{\hbox{-}} : t \sdot \blambda \overline \xi \sdot n \mapsto \pview{t \sdot \blambda \overline \xi \sdot n}$.}
\end{itemize}
\caption{Rules that define traversals over a long form \label{tab:traversals}}
\lo{N.B. By definition of binder function of a $\LamAt$-labelled binding tree, if the node $\alpha$ in $\lform{M}$ is labelled by a free variable, then $B(\alpha) = \epsilon$.}
\end{shadowbox}}
\end{center}
\end{table}

\begin{example}\label{eg:traversal}
There are three maximal traversals over the long form defined in Example~\ref{eg:longform1} (Figure~\ref{fig:lform-eg1}) as follows: 
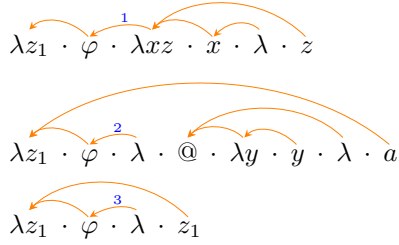
\[
\begin{array}{l}
\Pstr[0.65cm]{
(n1){\blambda z_1} \sdot 
(n2-n1){\phi} \sdot 
(n3-n2,30:1){\blambda x z} \sdot 
(n4-n3){x} \sdot 
(n5-n4){\blambda} \sdot 
(n6-n3){z}
} 
\\
\Pstr[0.65cm]{
(n1){\blambda z_1} \sdot 
(n2-n1){\phi} \sdot 
(n3-n2,30:2){\blambda} \sdot 
(n4){@} \sdot 
(n5-n4){\blambda y} \sdot 
(n6-n5){y} \sdot 
(n7-n4){\blambda} \sdot 
(n8-n1,40){a}
}
\\
\Pstr[0.65cm]{
(n1){\blambda z_1} \sdot 
(n2-n1){\phi} \sdot 
(n3-n2,30:3){\blambda} \sdot 
(n4-n1){z_1}
}
\end{array}
\]
\end{example}

\begin{lemma}
\label{lem:pview-path}
Let $t$ be a traversal over $\lform{M} = (T_M, B_M, \ell_M)$. Then $t$ is a well-defined justified sequence over $\ArExp\lform{M}$. Further
\begin{enumerate} 
\item The sequence underlying $\pview{t}$ is a path in the tree $T_M$, and the P-view determined by this path (Remark~\ref{rem:pview-path}) is exactly $\pview{t}$.
\item If $t$ is maximal then the last node of $t$ is labelled by a variable of ground type.
\item If $M$ is $\beta$-normal then $t$ is a P-view over $\ArExp\lform{M}$ (i.e.~$t = \pview{t}$). 
\end{enumerate}
\end{lemma}

\begin{proof}
(1) and (2) can be proved by a straightforward induction on the length of $t$; for (2), we appeal to the labelling axioms of Lemma~\ref{lem:bindingtree-char}.
For (3), since $\lform{M}$ does not have any $@$-labelled nodes, the traversal $t$ is constructed using only the rules (Root), (FVar) and (Lam). 
Thus $t$ is a path, which, thanks to Remark~\ref{rem:pview-path}, determines a P-view.
\end{proof}

Let $t \in \trmng{M}$ and $\Theta \subseteq \moves{\ArExp\lform{M}}$, define $\proj{t}{\Theta}$ to be the (justified) subsequence of $t$ consisting of nodes that are \changed[lo]{hereditarily justified} by some occurrence of an element of $\Theta$ in $t$.
\changed[lo]{If $\Theta$ is a set of initial moves, then} $\proj{\trmng{M}}{\Theta} := \makeset{\proj{t}{\Theta} \mid t \in \trmng{M}}$ is a well-defined set of justified sequences over $\ArExp\lform{M}$ (see e.g.~\citep[Lemma 2.6]{McCusker00}). 
Let $\alpha \in \moves{\ArExp\lform{M}}$, we write $\proj{\trmng{M}}{(\alpha, \Theta)}$ to mean $\proj{\trmng{M}}{(\makeset{\alpha} \cup \Theta)}$.

The rest of the section is about the following theorem and its proof. 
Given a term-in-context $\seq{\Gamma}{M : A}$ where $\Gamma = x_1 : C_1, \cdots, x_n : C_n$, recall that $\hat{\ell} : \ArExp\lform{M} \to \ArSuc\lform{M}$ is a direct arena morphism, and $\pview{\mng{\seq{\Gamma}{M : A}}}$ is a set of justified sequences over the arena $\mng{C_1 \to \cdots \to C_n \to A}$ and hence over the succinct arena $\ArSuc\lform{M}$ (the former is a subarena of the latter).

\begin{theorem}[Strong Bijection]\label{thm:corr}
Let $\seq{\Gamma}{M : A}$ be a term-in-context, with long form $\lform{M} = (T, B, \ell)$. 
The extension
\(
{\hat{\ell}^\ast} : \proj{\trmng{M}}{\epsilon}\; 
\stackrel{\sim}{\longrightarrow} 
\; \pview{\mng{\seq{\Gamma}{M : A}}}
\)
induced by the direct arena morphism $\hat{\ell}$ is a strong bijection.
\end{theorem}

In general $\hat{\ell}$ is neither injective nor surjective;
nevertheless $\hat{\ell}^\ast$ defines a bijection from $\proj{\trmng{M}}{\epsilon}$ to $\pview{\mng{\seq{\Gamma}{M : A}}}$, which is \emph{strong},
in the sense that for each $t \in \proj{\trmng{M}}{\epsilon}$, the two justified sequences $t$ and $\hat{\ell}^\ast(t)$ are isomorphic.




\begin{example}\label{eg:theorem}
To illustrate Theorem~\ref{thm:corr}, consider the term-in-context of Example~\ref{eg:longform1}, $\seq{\Gamma}{M : (o, o)}$, where $\Gamma = \makeset{\phi : (((o, o), o, o), o, o, o), \; a : o}$ and
\[\lform{M} = \lterm{z_1^o}{\phi \, (
\lterm{x^{(o, o)} z^o}{x \, (\lterm{}{z})}
) \, (\lterm{}{@ \, (\lterm{y^o}{y}) \, (\lterm{}{a})}) \, 
(\lterm{}{ z_1 })}
\]
The interpretaton, $\mng{\seq{\Gamma}{M : (o, o)}}$, is an innocent strategy over the arena $\Ar(B)$ where $B = ((((o, o), o, o), o, o, o), o, o, o)$.
Take a $\Lam$-representation of $\Ar(B)$ as follows:
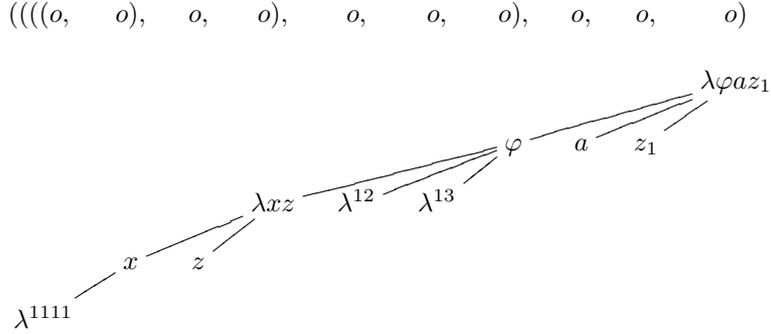
\[
\xymatrix@R-.5cm@C-.5cm{
((((o, & o), & o, & o), & o, & o, & o), & o, & o, & o)\\
 &&& &&& &&& \blambda \phi a z_1 \ar@{-}[llld] \ar@{-}[lld] \ar@{-}[ld] \\
 &&& &&& \phi \ar@{-}[llld] \ar@{-}[lld] \ar@{-}[ld]& a & z_1 & \\
 &&& \blambda x z \ar@{-}[lld] \ar@{-}[ld] & \blambda^{12} & \blambda^{13} & &&& \\
 & x \ar@{-}[ld] & z & &&& &&& \\\
\blambda^{1111} &&& &&& &&&
}
\]
Note that when restricted to $\proj{\trmng{M}}{\epsilon}$, the image of $\hat{\ell}^\ast$ consists of justified sequences over $\Ar(B)$, a subarena of $\ArSuc\lform{M}$.
Thus maximal justified sequences in $\hat{\ell}(\proj{\trmng{{M}}}{\epsilon}) $ are as follows (omitting the pointers) 
\[
\begin{array}{l}
	\blambda \phi a z_1 \sdot \phi \sdot \blambda x z \sdot x \sdot \blambda^{1111} \sdot z
	\\
\blambda \phi a z_1 \sdot \phi \sdot \blambda^{12} \sdot a
\\
\blambda \phi a z_1 \sdot \phi \sdot \blambda^{13} \sdot z_1
\end{array}
\]
coninciding with the maximal P-views in the strategy denotation $\mng{\seq{\Gamma}{M : (o, o)}}$. 
\end{example}

\subsection{Proof of the strong bijection theorem}

We first state and prove a useful lemma. 
Let $t$ be a traversal over the long form
\[{\lform{M} = \lterm{}{@ \; (\lterm{\overline\xi}{P}) \; (\lterm{\overline{\eta_1}}{Q_1}) \cdots (\lterm{\overline{\eta_n}}{Q_n})} : o }\]
Let $\theta_1, \cdots, \theta_l$ be a list of all the occurrences of the nodes, ${\blambda \overline\xi}, \blambda \overline{\eta_{1}}, \cdots, \blambda \overline{\eta_{n}}$, in $t$. Then, except for the first two nodes (i.e.~$\blambda$ and $@$) of the traversal $t$, 
every node occurrence $m$ in $t$ belongs to one of the \emph{components}, $\comp{t}{\theta_1}, \cdots, \comp{t}{\theta_l}$, defined as follows: 
\begin{itemize}
\item If $m$ is hereditarily justified by $\theta_i$ then $m$ belongs to $\comp{t}{\theta_i}$.
\item If $m$ is hereditarily justified by an \emph{internal} $@$ (as opposed to the top-level $@$), or by the root $\epsilon$ \changed[lo]{(i.e.~by a free variable)}, 
let $m'$ be the \emph{last} node occurrence in $t$ which precedes $m$ and which is hereditarily justified by $\theta_i$ for some $i$, then $m$ belongs to $\comp{t}{\theta_i}$.
\end{itemize}
Henceforth, by abuse of notation, by $\comp{t}{\theta_i}$ we mean the subsequence of $t$ determined by the set of node occurrences $\comp{t}{\theta_i}$.

\lo{N.B.
Suppose $m$ is justified by $m'$ in $t$. 
It follows from the definition that if $m$ belongs to $\comp{t}{\theta}$ and $m$ is not (an occurrence of) a free variable, then $m'$ also belongs to $\comp{t}{\theta}$. 
Hence, if $m$ is a free variable, then $m$ points to $\theta$, the opening node of $\comp{t}{\theta}$.
}

\lo{N.B. In general, $\comp{t}{\theta_i} \not= \proj{t}{\theta_i}$ (because the former, but not the latter, may contain nodes hereditarily justified by an internal $@$;
similarly $\comp{t}{\theta_i} \not= \proj{t}{(\theta_i, \epsilon)}$.}

\begin{lemma}[Component Projection]\label{lem:pviewproj}
Using the preceding notation, let $t$ be a traversal over the long form
\[{\lform{M} = \lterm{}{@ \; (\lterm{\overline\xi}{P}) \; (\lterm{\overline{\eta_1}}{Q_1}) \cdots (\lterm{\overline{\eta_n}}{Q_n})} : o }\]
If the last node of $t$ \changed[lo]{is a lambda node} which belongs to $\comp{t}{\theta}$ then  
\begin{enumerate}[(i)]
\item \(
\comp{t}{\theta} \; \in \; \trmng{{\theta . M}}
\) 
\item $\pview{t} = \blambda \sdot @ \sdot \pview{\comp{t}{\theta}}$
\end{enumerate}
where $\theta . M := \lterm{\overline \xi}{P}$ (respectively $\lterm{\overline{\eta_{j}}}{Q_{j}}$) in case $\theta$ is an occurrence of $\blambda \overline{\xi}$ (respectively $\blambda \overline{\eta_{j}}$).
\end{lemma}

\begin{proof}
\begin{asparaenum}[(i)]
\item Let $t$ be a traversal satisfying the premises of the lemma. 
It follows from the rules of Definition~\ref{def:traversal} that $t$ is a prefix of a justified sequence of the following shape:
\begin{equation}
t = \Pstr[0.7cm]{\blambda \sdot (a){@} \;\cdots\; (l-a){{\theta}} \sdot \; B_{1} \; \cdots \; B_{l}}
\label{eq:decomp_t}
\end{equation}
where ${\theta}$ is justified by $@$, and each $B_i$ is a block of nodes of one of two types:
\begin{itemize}
\item[I.] Two-node block $n_{i} \sdot l_{i}$ where the non-lambda node $n_{i}$ and the lambda node $l_{i}$ are hereditarily justified by an 
$@$ that belongs to $\comp{t}{\theta}$. 

\item[II.] 
$\Pstr[0.7cm]{(n){n_i} \; m_1 \cdots m_r \; (l-n){l_i}}$
where the lambda node $l_{i}$ is justified by $n_{i}$, which belongs to $\comp{t}{\theta}$ and is hereditarily justified by $ \theta$ or by $\epsilon$. 
\lo{N.B. In case $l_i$ and $n_i$ are hereditarily justified by $\epsilon$ then $r = 0$.}
\end{itemize}
It follows that for each $i$, both $n_i$ and $l_i$ are occurrences of nodes from the long form $\lform{\theta . M}$ \emph{qua} subtree of $\lform{M}$.
Now define $B_i'$ to be $n_i \sdot l_i$ for each $i$.
Observe that
\(
\comp{t}{\theta} = \theta \sdot B_1' \, \cdots \, B_l'.
\)
It then follows that $\comp{t}{\theta} \in \trmng{\seq{\Gamma}{\theta . M}}$.

\item Immediate consequence of (\ref{eq:decomp_t}) and
\(
\comp{t}{\theta} = \theta \sdot B_1' \, \cdots \, B_l'.
\)
\end{asparaenum}
\end{proof}

\lo{N.B. Thanks to the convention that free variables that occur in $\comp{t}{\theta}$ point to $\theta$, it follows that $\proj{\comp{t}{\theta}}{\theta} \in \proj{\trmng{\theta . M}}{\epsilon}$. 
}


\begin{lemma}
Let $\seq{\Gamma}{M:A}$ be in $\beta$-normal form with long form $\lform{M} = (T, B, \ell)$. 
The extension
$\hat{\ell}^\ast : \proj{\trmng{M}}{\epsilon} \; \stackrel{\sim}{\longrightarrow} \; \pview{\mng{\seq{\Gamma}{M : A}}}$ induced by $\hat{\ell}$ is a strong bijection.
\end{lemma}

\begin{proof}
Take the term-in-context $\seq{\Gamma}{M : o}$ where $\Gamma = x_1 : A_1, \cdots, x_n : A_n$, $A_i = (B_1, \cdots, B_m, o)$, 
and $M = x_i \, (\lterm{\overline{y_1}}{P_1}) \cdots (\lterm{\overline{y_m}}{P_m})$. 
We prove by induction on the size of $M$.
Observe that, since $M$ is $\beta$-normal, $\proj{\trmng{M}}{\epsilon} = \trmng{M}$.
First we show that the map $\hat\ell^\ast$ is injective. 
Let $\blambda \sdot x_i \sdot t, \blambda \sdot x_i \sdot t' \in \trmng{M}$ such that $\hat\ell^\ast(\blambda \sdot x_i \sdot t) = \hat\ell^\ast(\blambda \sdot x_i \sdot t')$, which implies that $\hat\ell^\ast(t) = \hat\ell^\ast(t')$, and $t, t' \in \trmng{\lterm{\overline{y_j}}{P_j}}$ for some $j \in [m]$. 
Since $\hat\ell^\ast : \trmng{\lterm{\overline{y_j}}{P_j}} \; \stackrel{\sim}{\longrightarrow} \; \pview{\mng{\seq{\Gamma}{\lterm{\overline{y_j}}{P_j} : B_j}}}$ is injective, we have $t = t'$ and hence $\blambda \sdot x_i \sdot t = \blambda \sdot x_i \sdot t'$ as desired.

For surjectivity of $\hat\ell^\ast$, take a P-view $p \in \pview{\mng{\seq{\Gamma}{M : o}}}$.
Notice that $p$ is a justified sequence over $\Ar(A_1, \cdots, A_n, o)$.
Since the head variable of $M$ is $x_i$, we have $p = \epsilon \sdot i \sdot (p' \uparrow i)$ and $p' \in \pview{\mng{\seq{\Gamma}{\lterm{\overline{y_j}}{P_j} : B_j}}}$ for some $j \in [m]$. (Given a sequence of nodes, $p = \alpha_1, \cdots, \alpha_n$, we write $p \uparrow i$ for the sequence $i \sdot \alpha_1, \cdots, i \sdot \alpha_n$.)
By the induction hypothesis, there exists $t \in \trmng{\lterm{\overline{y_j}}{P_j}}$ such that $\hat\ell^\ast(t) = p'$.
Thus we have $\blambda \sdot x_i \sdot t \in \trmng{M}$, and $\hat\ell^\ast(\blambda \sdot x_i \sdot t) = p$ as desired.
\end{proof}

\begin{corollary}[Variable] \label{cor:var-corrfinite}
Let $\lform{\phi} = (T, B, \ell)$ where $\phi$ is a variable. 
The extension
$\hat{\ell}^\ast : \proj{\trmng{\phi}}{\epsilon} \; \stackrel{\sim}{\longrightarrow} \; \pview{\mng{\seq{\Gamma}{\phi : A}}}$ induced by $\hat{\ell}$ is a strong bijection.
\end{corollary}



\begin{example}\rm
Consider the long form 
\[
\lform{\lterm{\chi}{\chi}} \; = \; \lterm{\chi \Phi \phi}{\chi \; (\lterm{\psi}{\Phi \; (\lterm{y}{\psi (\lterm{}{y})})})\; (\lterm{x}{\phi \; (\lterm{}{x})})}.
\] 
Since there is no occurrence of $@$ in $\lform{\lterm{\chi}{\chi}}$, traversals over it coincide with paths from the root. For example, the traversal $\blambda \chi \Phi \phi \sdot \chi \sdot \blambda \psi \sdot \Phi \sdot \blambda y \sdot \psi \sdot \blambda \sdot y$ (pointers are omitted) represents a P-view in the copycat strategy $\mng{\seq{}{\lterm{\chi}{\chi : A \to A}}}$. 
This illustrates the strong bijection of Lemma~\ref{cor:var-corrfinite}.
\end{example}

\lo{Should say what $\proj{(-)}{A}$ means in game semantics.}

\medskip

We recall the notion of interaction sequences and the associated notation from \citep{HO00}.
Let $\sigma : A \mor B$ and $\tau : B \mor C$ be innocent strategies, 
and let us write $\pview{\sigma}$ for the collection of P-views in $\sigma$.
Given a justified sequence $t$ over the triple $(A, B, C)$ of arenas,  \lo{Clarify definition of justified sequences over $(A, B, C)$} let $X$ range over the \emph{components} $(B, C)$ and $(A, B)_{b}$ where $b$ ranges over the occurrences of initial moves of $B$ in $t$; set 
\[
\rho_{X} \; := \;
\left\{
\begin{array}{ll}
\sigma & \hbox{if $X = (A, B)_{b}$}\\
\tau & \hbox{if $X = (B, C)$}
\end{array}
\right.
\]
Similarly we define $\pview{\rho_X}$ to mean $\pview{\sigma}$ or $\pview{\tau}$ depending on what $X$ is.
The set of \emph{interaction sequences between $\sigma$ and $\tau$}, $\intseq{\sigma, \tau}$, consists of justified sequences $t$ over $(A, B, C)$, which are defined by induction over the rules (IS1), (IS2) and (IS3):
\begin{itemize}
\item[(IS1)] $c \in \intseq{\sigma, \tau}$ where $c$ ranges over the initial moves of $C$.
\item[(IS2)] If $t \sdot m \in \intseq{\sigma, \tau}$, and $m$ is a \emph{generalised O-move} of the component $X$ (i.e.~either an O-move of $A \funsp C$ or a move of $B$), and $\pview{\proj{t \sdot m}{X}} \sdot m' \in \rho_{X}$, then $t \sdot m \sdot m' \in \intseq{\sigma, \tau}$.

\item[(IS3)] If $t \sdot m \in \intseq{\sigma, \tau}$, and $m$ is a P-move of $A \funsp C$, and $(\proj{t \sdot m}{(A, C)}) \sdot m'$ is a play of $A \funsp C$, then $t \sdot m \sdot m' \in \intseq{\sigma, \tau}$.
\end{itemize}

\begin{definition}\label{def:pvintseq} \rm
The set of \emph{P-visible interaction sequences between $\sigma$ and $\tau$}, $\intseqpv{\sigma, \tau}$ , consists of justified sequences over $(A, B, C)$, defined by induction over the rules (IS1), (IS2), and (IS4) as follows:

\begin{enumerate}
\item[(IS4)] If $t \sdot m \in \intseqpv{\sigma, \tau}$, and $m$ is a P-move of $A \funsp C$, and $m'$ is an O-move  justified by $m$, then $t \sdot m \sdot m' \in \intseqpv{\sigma, \tau}$.
\end{enumerate}
\end{definition}

It is straightforward to see that $\proj{\intseqpv{\sigma, \tau}}{(A, C)} = \pview{\sigma ; \tau}$.
Note that the definition would still make sense if $\sigma$ and $\tau$ in (IS1), (IS2) and (IS4) are replaced by $\pview{\sigma}$ and $\pview{\tau}$ respectively, and $\rho_X$ replaced by $\pview{\rho_X}$. 
In other words~$\intseqpv{\pview{\sigma}, \pview{\tau}}$ is well-defined, and coincides with $\intseqpv{\sigma, \tau}$.
Thus we have $\proj{\intseqpv{\pview{\sigma}, \pview{\tau}}}{(A, C)} = \pview{\sigma ; \tau}$.


\begin{lemma}\label{lem:corr}
\begin{enumerate}[(i)]
\item Let $\seq{\Gamma}{M : A}$ be a term-in-context with long form $\lform{M} = (T, B, \ell)$. There is a $\hat{\ell}$-induced strong bijection 
\(
\hat{\ell}^\ast \; : \; \proj{\trmng{M}}{\epsilon}\; \stackrel{\sim}{\longrightarrow} \; \pview{\mng{\seq{\Gamma}{M : A}}}.
\)

\item \label{lem:corrfinite}
Suppose \(\seq{\Gamma}{\lform{M} = \lterm{}{@ \; (\lterm{\overline\xi}{P}) \; (\lterm{\overline{\eta_1}}{Q_1}) \cdots (\lterm{\overline{\eta_n}}{Q_n})} : o }\), and $\lform{M} = (T, B, \ell)$.
There is a $\hat{\ell}$-induced 
bijection.
\begin{equation}
{\hat{\ell}^\ast} \; : \; \proj{\trmng{{M}
}}{(\epsilon, \Theta)} \; \stackrel{\sim}{\longrightarrow} \; 
\intseqpv{\anglebra{p, q_1, \cdots, q_n}, \pview{\mathit{ev}}} \label{eq:ICorr1}
\end{equation} 
where $\Theta := \makeset{\alpha \in T \mid \exists i \, . \, \alpha = 1 \sdot i }$, $p = \pview{\mng{\seq{\Gamma}{\lterm{\overline \xi}{P} : (\prod_{i=1}^n B_i) \funsp o}}}$, 
$q_{i} = \pview{\mng{\seq{\Gamma}{\lterm{\overline{\eta_i}}{Q_i} : B_i }}}$ for each $i$, and 
\(\xymatrix{ ((\prod_{i=1}^n B_i) \funsp o) \times \prod_{i=1}^n B_i \ar[r] ^-{\mathit{ev}} & o}\) is the obvious copycat strategy.
\end{enumerate}
\end{lemma}

\begin{remark}\label{rem:bijection}
In (ii), since $\ell(1) = @$, each $t \in \proj{\trmng{{M} }}{(\epsilon, \Theta)}$ is a subsequence of a traversal over $\lform{M}$ consisting of nodes that are hereditarily justified by $\epsilon$, 
or by an occurrence of one of $\blambda \overline\xi, \blambda \overline{\eta_1}, \cdots, \blambda \overline{\eta_n}$ (each being a child of node 1 in the tree $\lform{M}$).
The bijection $\hat{\ell}$ in (ii) would be strong if for each $t \in \proj{\trmng{{M} }}{(\epsilon, \Theta)}$, pointers were added from every occurrence of $\blambda \overline \xi, \blambda \overline{\eta_1}, \cdots, \blambda \overline{\eta_n}$ to the opening node $\epsilon$.
\end{remark}

\begin{proof} We shall prove (i) and (ii) by mutual induction.

\medskip

\begin{asparaenum}[(i)]
\item
The term $M$ has one of the following shapes: 
\begin{enumerate}[(a)]
\item abstraction $\lterm{\xi}{P}$
\item variable $\phi$
\item application $N \, L_1 \cdots L_n$ where $n \geq 1$ and $N $ has shape (a) or (b).
\end{enumerate}

\medskip

First we reduce case (a) to case (b) or case (c). 
Let $A = (A_1, \cdots, A_n, o)$ and the $\eta$-long normal form of $\lterm{\xi}{P}$ be $\lterm{x_1 \cdots x_n}{R}$ where $R$ is a term of either case (b) or (c). 
Plainly $\trmng{{\lterm{\xi}{P}}} = \trmng{{\lterm{\overline x}{R}}}$. 
It remains to observe that, on the one hand, there is a strong bijection between $\trmng{{\lterm{\overline x}{R}}}$ and 
$\trmng{{{R}}}$ (note that pointers from node occurrences labelled with free variables are to the opening node);
and on the other, there is a strong bijection between $\pview{\mng{\seq{\Gamma}{\lterm{\overline x}{R : A}}}}$ and $\pview{\mng{\seq{\Gamma, \overline x : \overline A}{{R} : o}}}$. 

\medskip

Case (b) is just Corollary~\ref{cor:var-corrfinite}. 

\medskip

As for case (c), suppose $N$ is an abstraction. 
W.l.o.g.~assume
\[
{\lform{M} = \lterm{}{@ \; (\lterm{\overline\xi}{P}) \; (\lterm{\overline{\eta_1}}{Q_1}) \cdots (\lterm{\overline{\eta_n}}{Q_n})} : o}.
\] 
Then, by (ii) and using the notation therein, we have a bijection
\begin{equation}
\hat{\ell} \; : \; \proj{\trmng{{M}
}}{(\epsilon, \Theta)} \; \stackrel{\sim}{\longrightarrow} \; 
\intseqpv{\anglebra{p, q_1, \cdots, q_n}, \pview{\mathit{ev}}} 
\label{eq:ICorr2}
\end{equation}
Observe that applying $\proj{(-)}{\epsilon}$ on the LHS of (\ref{eq:ICorr2}) corresponds to applying $\proj{(-)}{(\mng{\Gamma}, o)}$ on the RHS.
Thus, in view of Remark~\ref{rem:bijection}, $\proj{\proj{\trmng{{M}
}}{(\epsilon, \Theta)}}{\epsilon} = \proj{\trmng{{M}}}{\epsilon}$ is in strong bijection with
$\proj{\intseqpv{\anglebra{p, q_1, \cdots, q_n}, \pview{\mathit{ev}}}}{(\mng{\Gamma}, o)} = \pview{\mng{\seq{\Gamma}{M : o}}}$, as desired.

\medskip

Finally suppose $N$ is a variable $\phi$. 
W.l.o.g.~assume $M = \phi (\lterm{\overline{\eta_1}}{Q_1}) \cdots (\lterm{\overline{\eta_m}}{Q_m}) : o$.
Then it follows from the respective definitions that
\begin{align*}
\trmng{{M}} &= 
\bigcup_i 
\makeset{
\blambda \sdot \phi \sdot t
\mid t \in 
\trmng{
\seq{\Gamma}{
\lterm{\overline{\eta_i}}{Q_i}
}
}
} 
\cup \makeset{\epsilon, \blambda}
\\
\pview{\mng{\seq{\Gamma}{M}}} &= 
\bigcup_i 
\makeset{
\hat{\ell}(\blambda) \sdot \hat{\ell}(\phi) \sdot t
\mid
t \in \pview{\mng{\seq{\Gamma}{\lterm{\overline{\eta_i}}{Q_i}}}} 
}
\cup \makeset{\epsilon, \hat{\ell}(\blambda)}
\end{align*}
It follows from the induction hypothesis that there is a strong bijection
\[
\hat{\ell} \; : \; \proj{\trmng{{M}}}{\epsilon}\; \stackrel{\sim}{\longrightarrow} \; \pview{\mng{\seq{\Gamma}{M : o}}}.
\]

\item
We first show that for every $t \in \trmng{{ M}}$, we have ${\hat{\ell}}(\proj{t}{(\epsilon, \Theta)}) \in \intseqpv{\anglebra{p, q_1, \cdots, q_n}, \pview{\mathit{ev}}}$. 
The proof is by induction on the length of $t$, with case distinction on the last node of $t$, using the notation of Lemma~\ref{lem:pviewproj}.

\medskip

\emph{Case 1}. The last node of $t$ is in the component $\comp{t}{\theta}$ where $\theta$ is the unique occurrence of $\blambda \overline \xi$ in $t$.
Let ${m}$ be the last lambda node in $t$ that is hereditarily justified by $\epsilon$ or $\theta$. There are two subcases.
\begin{itemize}

\item[\emph{Case 1.1}.] \lo{(IS2)-$\sigma$} The traversal $t = \cdots {m} \sdot d_{1} \, d_2 \cdots d_l$ where $l \geq 0$ and each $d_{i}$ is a Type-I node (i.e.~not hereditarily justified by $\epsilon$ or $\theta$), as defined in the proof of Lemma~\ref{lem:pviewproj}.
Let $m'$ be a Type-II node (i.e.~hereditarily justified by $\epsilon$ or $\theta$) such that $t \sdot d_{l+1} \cdots d_{l'} \sdot m' \in \trmng{{M}}$ where $l' \geq l$ and each $d_{j}$ is Type I. 
Writing $t' = t \sdot d_{n+1} \cdots d_{n'}$,
we claim that ${\hat{\ell}}(\proj{(t' \sdot m')}{(\epsilon, \Theta)}) = {\hat{\ell}}(\proj{t'}{(\epsilon, \Theta)}) \sdot {\hat{\ell}}(m') \in \intseqpv{\anglebra{p, \overline q}, \pview{\mathit{ev}}}$ by rule (IS2)-$\sigma$.
By assumption ${\hat{\ell}}(m')$, which is a P-move, belongs to the component
\[
(\mng{\Gamma}, (\prod_{i=1}^n B_i \funsp o) \times \prod_{i=1}^n B_i)_{o} = (\mng{\Gamma}, \prod_{i=1}^n B_i\funsp o)_{o}
\]
in the sense of the strategy composition $\mng{\Gamma} \xrightarrow{ \anglebra{p, \overline q} } (\prod_{i=1}^n B_i \funsp o) \times \prod_{i=1}^n B_i \xrightarrow{\mathit{ev}} o$. 
Observe that 
\begin{enumerate}[(A)]
\item $\proj{{\hat{\ell}}(\proj{t'}{(\epsilon, \Theta)})}{(\mng{\Gamma}, \prod_{i=1}^n B_i \funsp o)_o} = {\hat{\ell}}(\proj{\proj{t'}{(\epsilon, \Theta)}}{\theta})$,
\item $\proj{\proj{t'}{(\epsilon, \Theta)}}{\theta} = \proj{t'}{\theta} = \proj{\comp{t'}{\theta}}{\theta}$, because all node occurrences in $t'$ hereditarily justified by $\theta$ are in $\comp{t'}{\theta}$.
\end{enumerate}
Thus we have
\[
\begin{array}{lll}
& \pview{\proj{{\hat{\ell}}(\proj{t'}{(\epsilon, \Theta)})}{(\mng{\Gamma}, \prod_{i=1}^n B_i \funsp o)}} \sdot {\hat{\ell}}(m') & \\
= & & \hbox{(A)}\\
& \pview{{\hat{\ell}}(\proj{\proj{t'}{(\epsilon, \Theta)}}{\theta})} \sdot {\hat{\ell}}(m') & \\
= & & \hbox{(B)}\\
& \pview{{\hat{\ell}}(\proj{\comp{t'}{\theta}}{\theta})} \sdot {\hat{\ell}}(m') & \\
= & & \hbox{$m'$ is a non-lambda node}\\
 & \pview{{\hat{\ell}}(\proj{\comp{t' \sdot m'}{\theta}}{\theta})} & \\
\in & & \hbox{Lemma~\ref{lem:pviewproj}(i) \& I.H.(i)}\\
 & \pview{\mng{\seq{\Gamma}{\lterm{\overline \xi}{P}}}}
\end{array}
\]
as desired.

\item[\emph{Case 1.2}.] The traversal $ t = \cdots m \sdot d_{1} \cdots d_{l} \sdot m'$ where $l \geq 0$ and each $d_{i}$ is of Type I (i.e.~not hereditarily justified by $\epsilon$ or $\theta$), and $m'$ is a non-lambda node. 

There are two subcases.
\begin{itemize}

\item[\emph{Case 1.2.1}.] \lo{(IS4)} The non-lambda node $m'$ is hereditarily justified by $\epsilon$. 

By rule (FVar) of the definition of traversal, for each lambda node $m''$ that is  justified by $m'$, we have $t \sdot m'' \in \trmng{{M}}$. 
We claim that ${\hat{\ell}}(\proj{t \sdot m''}{(\epsilon, \Theta)}) \in \intseqpv{\anglebra{p, \overline{q}}, \pview{\mathit{ev}}}$. 
By the induction hypothesis, $\proj{t}{(\epsilon, \Theta)} = (\proj{\trunc{t}{d_l}}{(\epsilon, \Theta)}) \sdot m'$ is mapped by ${\hat{\ell}}$ into $\intseqpv{\anglebra{p, \overline{q}}, \pview{\mathit{ev}}}$. 
Since ${\hat{\ell}}(m')$ is an P-move of the arena $\mng{\Gamma} \funsp o$, 
it follows from rule (IS4) of the definition of P-visible interaction sequence that ${\hat{\ell}}(\proj{\trunc{t}{d_l}}{(\epsilon, \Theta)} \sdot m') \sdot {\hat{\ell}}(m'') =  {\hat{\ell}}(\proj{t \sdot m''}{(\epsilon, \Theta)}) \in \intseqpv{\anglebra{p, \overline{q}}, \pview{\mathit{ev}}}$ as required.

\item[\emph{Case 1.2.2}.] \lo{(IS2)-$\tau$} The non-lambda node $m'$ is hereditarily justified by $\theta$. 

Suppose $t = \cdots \underline{m}'' \sdot \underline{m}' \cdots m \sdot \overline d \sdot m'$
such that $m'$ is explicitly $i$-justified by the lambda node $\underline{m}'$. 
By rule (BVar) of the definition of traversal, $\underline{m}''$ is a non-lambda node in $\comp{t}{\theta'}$ where $\theta'$ is an occurrence of some $\blambda \overline{\eta_j}$ in $t$, and $t \sdot m'' \in \trmng{{M}}$ where $m''$ is  $i$-justified by $\underline{m}''$. 
We claim that $\proj{t \sdot m''}{(\epsilon, \Theta)} = (\proj{t}{(\epsilon, \Theta)}) \sdot m''$ is mapped by ${\hat{\ell}}$ into $\intseqpv{\anglebra{p, \overline{q}}, \pview{\mathit{ev}}}$. 
Writing $\widetilde{t} =  \proj{t}{(\epsilon, \Theta)}$, by the induction hypothesis, we have ${\hat{\ell}}(\widetilde{t}) \in \intseqpv{\anglebra{p, \overline{q}}, \pview{\mathit{ev}}}$. 
By rule (IS2)-$\tau$ of the definition of interaction sequence, it suffices to show $\pview{\proj{{\hat{\ell}}(\widetilde{t})}{((\prod_{i=1}^n B_i \funsp o) \times \prod_{i=1}^n B_i, o)}}  \sdot {\hat{\ell}}(m'') \in \pview{\mathit{ev}}$. 
Since projecting to $((\prod_{i=1}^n B_i \funsp o) \times \prod_{i=1}^n B_i, o)$ reverses the P/O polarity of moves, we have \[
\pview{\proj{{\hat{\ell}}(\widetilde{t})}{((\prod_{i=1}^n B_i \funsp o) \times \prod_{i=1}^n B_i, o)} }  \sdot {\hat{\ell}}(m'') = 
\Pstr[0.7cm]{\cdots \sdot (n){{\hat{\ell}}(\underline{m}'')} \sdot (lxi){{\hat{\ell}}(\underline{m}')} \sdot (xi-lxi,40:i){{\hat{\ell}}(m')} \sdot (leta-n,35:i){{\hat{\ell}}(m'')}}
\] which is in $\pview{\mathit{ev}}$ as required.
\end{itemize}
\end{itemize}

It remains to show that for every $u \in \intseqpv{\anglebra{p, \overline{q}}, \pview{\mathit{ev}}}$, 
there exists a unique $t_u \in \trmng{{M}}$ such that ${\hat{\ell}}(\proj{t_u}{(\epsilon, \Theta)}) = u$. 

We argue by induction on the length of $u$. 
The base case of (IS1) is trivial.
For the inductive case, suppose $u \sdot n \in \intseqpv{\anglebra{p, \overline{q}}, \pview{\mathit{ev}}}$.
With reference to Definition~\ref{def:pvintseq}, there are three cases, namely, (IS2)-$\sigma$, (IS2)-$\tau$ and (IS4), which correspond to the preceding cases of 1.1, 1.2.2 and 1.2.1 respectively.
Here we consider the case of (IS2)-$\sigma$ for illustration; the other cases are similar and simpler.
I.e.~by assumption, we have $\pview{\proj{u}{(\mng{\Gamma}, \prod_{i=1}^n B_i \funsp o)_o}} \sdot n \in \pview{\mng{\seq{\Gamma}{\lterm{\overline \xi}{P}}}}$ where $n$ is a P-move of $\mng{\seq{\Gamma}{\lterm{\overline \xi}{P}}}$.
By the induction hypothesis, there exists a unique maximal $t_u \in \trmng{{M}}$ such that ${\hat{\ell}}(\proj{t_u}{(\epsilon, \Theta)}) = u$.
Since ${\hat{\ell}}(\proj{\comp{t_u}{\theta}}{\theta}) = \proj{{\hat{\ell}}(\proj{t_u}{(\epsilon, \Theta)})}{(\mng{\Gamma}, \prod_{i=1}^n B_i \funsp o)_o}$ (which is (B) of Case 1.1),
we have $\pview{\proj{{\hat{\ell}}(\comp{t_u}{\theta}}{\theta})} \sdot {\hat{\ell}}(m') \in \pview{\mng{\seq{\Gamma}{
\lterm{\overline \xi}{P}
}}}$ where $n = {\hat{\ell}}(m')$ for some $m'$.
Thanks to the strong bijection of (i), we have $\pview{\proj{\comp{t_u}{\theta}}{\theta} \sdot m'} \in \proj{\trmng{{\lterm{\overline \xi}{P}}}}{\epsilon}$ for a unique $m'$.
Then, because $\comp{t_u}{\theta} \in \trmng{{\lterm{\overline \xi}{P}}}$ by Lemma~\ref{lem:pviewproj}(i) and because the last node of $t_u$ is a lambda node, we have $\comp{t_u}{\theta} \sdot m' \in \trmng{{\lterm{\overline \xi}{P}}}$. 
By rule (Lam) of the definition of traversals, we have the $\pview{\comp{t_u}{\theta}} \sdot m'$ is a path in the tree $\lform{\lterm{\overline \xi}{P}}$.
It follows that $\blambda \sdot @ \sdot \pview{\comp{t_u}{\theta} \sdot m'}$ is a path in the tree $\lform{M}$.
But, by Lemma~\ref{lem:pviewproj}(ii), $\pview{t_u} \sdot m' = \blambda \sdot @ \sdot \pview{\comp{t_u}{\theta}} \sdot m'$. Hence, by rule (Lam), $t_u \sdot m' \in \trmng{{M}}$ with ${\hat{\ell}}(\proj{t_u \sdot m'}{(\epsilon, \Theta)}) = u \sdot n$ as desired.

\medskip

\emph{Case 2}. The last node of $t$ belongs to $\comp{t}{\theta'}$ where $\theta'$ is an occurrence of $\blambda\overline{\eta_i}$ in $t$.
This case is symmetrical to Case 1.
\end{asparaenum}
\end{proof}

\newcommand\traversal[1]{{\sf tr}(#1)}

\section{Application}






\subsection{Interpreting higher-order recursion schemes}
We assume the standard notion of higher-order recursion scheme \citep{KnapikNU02,Ong06}.
Fix a (possibly infinite) higher-order recursion scheme $\calG = \roundbra{\Sigma, {\cal N}, {\cal R}, F_{1}}$ over a ranked alphabet $\Sigma = \makeset{a_{1} : r_{1}, \ldots, a_{l} : r_{l}}$ where $r_i$ is the arity of the terminal $a_{i}$; with non-terminals ${\cal N} = \ong{\makeset{F_i :  A_i \mid i \in {\cal I}}}$ 
and rules $F_{i} \to \lterm{\overline{x_{i}}}{M_{i}}$ for each \ong{$i \in {\cal I}$}, and $F_{1} : o$ is the start symbol. 
Note that we do not assume $\cal I$ to be finite.
\ong{Henceforth we assume $\cal I = \omega$ for convenience, and regard $\Sigma$ as a set of free variables.}

We first give the semantics of $\calG$. 
Writing $\mng{\Sigma} := \prod_{i = 1}^{l} \mng{o^{r_{i}} \to o}$ and $\mng{{\cal N}} := \prod_{\ong{i \in \omega}} \mng{A_{i}}$, 
the semantics of $\calG$, $\mng{\seq{\Sigma}{F_1 : o}} : \mng{\Sigma} \longrightarrow \mng{o}$, is the composite
\[ \mng{\Sigma} 
\xrightarrow{\Lambda({\mathbf g})}
(\mng{{\cal N}} \Rightarrow \mng{{\cal N}}) 
\xrightarrow{{\cal Y}_{\mng{\cal N}}}
\mng{{\cal N}} 
\xrightarrow{\pi_1}
\mng{o}\]
in the category ${\mathbb I}$ of arenas and innocent strategies, where
\begin{itemize}
\item ${\mathbf g} : \mng{\Sigma} \times \mng{{\cal N}} \longrightarrow \mng{{\cal N}}$ is $\mng{ {\Sigma, {\cal N}} \vdash {
\roundbra{\lterm{\overline{x_{1}}}{M_{1}}, \cdots,
\lterm{\overline{x_{m}}}{M_{m}}, \ong{\cdots}}
: \prod_{\ong{i \in\omega}} A_{i}
} }$, and $\Lambda(\hbox{-})$ is currying
\item ${\cal Y}_{\ong{A}} : (A \funsp A) \to A$ is the fixpoint strategy (see~\citep[\S 7.2]{HO00}) \ong{over an arena $A$}
\item $\pi_{1}$ is the projection map. 
\end{itemize}

\begin{remark}\rm\label{rem:gsemHORS}
Since $\pview{\mng{\seq{\Sigma}{F_1:o}}}$ coincide with the branch language\footnote{Let $m$ be the maximum arity of the $\Sigma$-symbols, and write $[m] := \makeset{1, \cdots, m}$. The \emph{branch language} of
  $t : \mathit{dom}(t) \longrightarrow \Sigma$ consists of
  \begin{inparaenum}[(i)]
  \item infinite words $(f_1, d_1)(f_2, d_2) \cdots $ such that there exists $d_1 \, d_2 \cdots \in {[m]}^\omega$ such that $t(d_1
  \cdots d_i) = f_{i+1}$ for every $i \in \omega$ and
  \item finite words $(f_1, d_1) \cdots (f_n,  d_n) \, f_{n+1}$ such that there exists $d_1 \cdots d_n \in {[m]}^\star$ such that $t(d_1 \cdots d_i) =
  f_{i+1}$ for $0 \leq i \leq n$, and the arity of \(f_{n+1}\) is \(0\).
  \end{inparaenum}
} of the $\Sigma$-labelled tree generated by $\calG$, we can identify $\pview{\mng{\seq{\Sigma}{F_1:o}}}$ with $\mng{\calG}$, the tree generated by $\calG$.
\end{remark}

\subsection{The traversal-path correspondence theorem}


\newcommand\myf{{\mathbb F}}
\newcommand\Gapprox[1]{\textbf{G}^{\langle #1 \rangle}}
\newcommand\Yapprox[2]{{Y}^{\langle #1 \rangle}_{#2}}
\newcommand\Ghors[1]{{G}^{\langle #1 \rangle}}

Fix a higher-order \ong{infinite} recursion scheme $\calG = \roundbra{\Sigma, {\cal N}, {\cal R}, F_{1}}$, using the same notation as before. 
Define an $\omega$-indexed family of $\lambda$-terms, $\Gapprox{n} : \prod_{\ong{i \in \omega}} A_{i}$ with $n$ ranging over $\omega$, as follows:
\begin{align*}
\Gapprox{0} & := \roundbra{\lterm{\overline{x_{1}}}{\bot^{o}}, \cdots, \lterm{\overline{x_{m}}}{\bot^{o}}, \ong{\cdots}}\\
\Gapprox{n+1} & := \roundbra{
\lterm{\overline{x_{1}}}{M_{1}[\Gapprox{n} / \overline{F}]}, \cdots,
\lterm{\overline{x_{m}}}{M_{m}[\Gapprox{n} / \overline{F}], \ong{\cdots}}
}
\end{align*}
where $\bot^{A}$ is a constant symbol of type $A$, and $(\hbox{-})[\Gapprox{n} / \overline{F}]$ means the simultaneous substitution $(\hbox{-})[\pi_{1} \, \Gapprox{n} / F_{1}, \cdots, \pi_{m} \,\Gapprox{n} / F_{m}, \ong{\cdots} ]$, and $\pi_i \, (s_1, s_2, \cdots)$ is a short hand for $s_i$.
Write $\Ghors{n} := \pi_{1 } \, \Gapprox{n}$ for each $n \in \omega$. Note that each $\Ghors{n}$ is a (recursion-free) $\lambda$-term of type $o$.

\begin{lemma}\label{lem:gsemcont}
$\mng{\seq{\Sigma}{F_1 : o}} = \lub_{n \in \omega} \mng{\seq{\Sigma}{\Ghors{n} : o}}$.
\end{lemma}

\begin{proof}
Because $\mathbb A$ is a CCC that is enriched over the CPOs, for each type $A$, there is a fixpoint strategy ${\cal Y}_{A} : (A \funsp A) \funsp A$ 
which is the least (with respect to the enriching order, namely, set inclusion) fixpoint of the map $(A \funsp A) \funsp A \longrightarrow (A \funsp A) \funsp A$ that is the denotation of the $\lambda$-term 
\[\lterm{F : (A \funsp A) \funsp A}{\lterm{f : A \funsp A}{f \; (F \; f)}} \]
Define the following family of $\lambda$-terms, $\Yapprox{i}{A} : (A \funsp A) \funsp A$ with $i \in \omega$
\[\begin{array}{rll}
\Yapprox{0}{A} & := & \lterm{f : A \funsp A}{\bot^A}\\
\Yapprox{n+1}{A} & := & \lterm{f : A \funsp A}{f \; (\Yapprox{n}{A} \; f)}
\end{array}\] 
Note that $\Gapprox{n} = \Yapprox{n}{\cal N} \; \roundbra{
\lterm{\overline{x_{1}}}{M_{1}}, \cdots, 
\lterm{\overline{x_{m}}}{M_{m}}, \ong{\cdots}} : {\cal N}$, 
writing ${\cal N} = \prod_{\ong{i \in \omega}} A_i$ by abuse of notation. 
By interpreting $\bot^A$ as the least element of the homset ${\mathbb A}_\Sigma(\textbf{1}, A)$, 
we have $\mng{\seq{\Sigma}{\Ghors{n} : o}} = \Lambda({\bf g}) ; \mng{ \Yapprox{n}{\cal N}} ; \pi_1$. 
By Knaster-Tarski Fixpoint Theorem, ${\cal Y}_A = \lub_{n \in \omega} \mng{ \Yapprox{n}{A} }$. 
Since composition is continuous, we have $\mng{\seq{\Sigma}{F_1 : o}} = \Lambda({\bf g}) ; {\cal Y}_{\mng{\cal N}} ; \pi_1 = \Lambda({\bf g}) ; \lub_{n \in \omega} \mng{ \Yapprox{n}{\cal N} } ; \pi_1 = \lub_{n \in \omega} (\Lambda({\bf g}) ; \mng{ \Yapprox{n}{\cal N} } ; \pi_1) = \lub_{n \in \omega} \mng{\seq{\Sigma}{\Ghors{n} : o}}$ as desired.
\end{proof}

\begin{lemma}[P-view Decomposition]\label{lem:decomppview}
\begin{enumerate}[(i)]
\item For every (possibly infinite) P-view $p = \lub_{i \in \omega} p_{i} \in \mng{\seq{\Sigma}{F_1 : o}}$ where each $p_{i}$ is a finite P-view such that $p_{0} \leq p_{1} \leq p_{2} \leq \cdots$, there is an increasing sequence of natural numbers $n_{0} < n_{1} < n_{2} < \cdots$ such that each $p_{i} \in \mng{\seq{\Sigma}{\Ghors{n_{i}} : o}}$. 
\item For every $\omega$-indexed family of finite P-views, $p_{i} \in \mng{\seq{\Sigma}{\Ghors{n_{i}} : o}}$ with $i \in \omega$, such that $p_{0} \leq p_{1} \leq p_{2} \leq \cdots$ and an infinite sequence of natural numbers $n_{0} < n_{1} < n_{2} < \cdots$, the (possibly infinite) P-view $\lub_{i \in \omega} p_{i} \in \mng{\seq{\Sigma}{F_1 : o}}$.
\end{enumerate}
\end{lemma}
 
 \begin{proof}
 \begin{inparaenum}[(i)]
\item Thanks to Lemma~\ref{lem:gsemcont}, for every $i$, there exists $n_i \geq 1$ such that $p_i \in \mng{\seq{\Sigma}{\Gapprox{n_i} : o}}$.
\item An immediate consequence of Lemma~\ref{lem:gsemcont}.
 \end{inparaenum}
 \end{proof}
 
Given a higher-order recursion scheme $\calG$, the \emph{computation tree} $\lambda(\calG)$ is obtained by first transforming the rewrite rules into long forms, and then unfolding the transformed rules \emph{ad infinitum}, starting from $F_1$, and without performing any $\beta$-reduction (i.e.~substitution of actual parameters for formal parameters); see \citep{Ong06} for a definition.
By construction, the tree $\lambda(\calG)$ is a (possibly infinite) $\LamAt$-labelled binding tree that satisfies the labelling axioms (Lam), (Leaf), (TVar) and (T@).
We write $\trmng{\calG}$ be the set of finite and infinite traversals over $\lambda(\calG)$, whereby an infinite traversal is just an infinite justified sequence over $\lambda(\calG)$ such that every finite prefix is a traversal. 

Next we prove a similar decomposition lemma for traversals. 
First, notice that each $\lform{\Ghors{i}}$ is a $\LamAt_\bot$-labelled binding tree, where $\LamAt_\bot$ is $\LamAt$ augmented by an auxiliary symbol $\bot$ of arity 0. 
Given $\LamAt_\bot$-labelled trees $T$ and $T'$, we define $T \sqsubseteq T'$ if $\dom(T) \subseteq \dom(T')$, and for all $\alpha \in \dom(T)$, if $T(\alpha) \not= \bot$ then $T(\alpha) = T'(\alpha)$.
Thus if $i < j$ then $\lform{\Ghors{i}} \sqsubseteq \lform{\Ghors{j}}$ 
and $\lform{\Ghors{i}} \sqsubseteq \lform{G}$. 
It follows that if $t$ is a $\bot$-free traversal over $\lform{\Ghors{i}}$ then $t$ is also a traversal over $\lform{\Ghors{j}}$ and over $\lform{G}$. 
Conversely if $t$ is a finite traversal over $\lform{G}$, then for every $n$ greater than the length of $t$, $t$ is also a traversal over $\lform{\Ghors{n}}$. 
To summarise, we have the following.


\begin{lemma}[Traversal Decomposition]\label{lem:decomptrav}
\begin{enumerate}[(i)]
\item For every (possibly infinite) traversal $t = \lub_{i \in \omega} t_{i} \in \trmng{\calG}$ where each $t_{i}$ is a finite traversal and $t_{0} \leq t_{1} \leq t_{2} \leq \cdots$, there is an increasing sequence of natural numbers $n_{0} < n_{1} < n_{2} < \cdots$ such that each $t_{i} \in \trmng{{\Ghors{n_{i}} : o}}$. 
\item Given an $\omega$-indexed family of $\bot$-free traversals, $ t_{i} \in \trmng{{\Ghors{n_{i}} : o}}$ with $i \in \omega$, such that $t_{0} \leq t_{1} \leq t_{2} \leq \cdots$, and an infinite sequence of natural numbers $n_{0} < n_{1} < n_{2} < \cdots$, the (possibly infinite) traversal $\lub_{i \in \omega} t_{i} \in \trmng{{\calG}}$.
\myendproof
\end{enumerate}
\end{lemma}

\begin{theorem}[Traversal-Path Correspondence]\label{thm:11traversal}
\ong{Let $\calG = \roundbra{\Sigma, {\cal N}, {\cal R}, F_{1}}$ be a possibly infinite higher-order recursion scheme.} 
Paths in $\mng{\calG}$ and traversals over $\lambda(\calG)$ projected to symbols from $\Sigma$ are the same set of finite and infinite sequences over $\Sigma$.
\endproof
\end{theorem}

\begin{proof}
By combining Lemma~\ref{lem:decomppview}, Lemma~\ref{lem:decomptrav} and Theorem~\ref{thm:corr}, we obtain a bijection
\(
\phi : \proj{\trmng{\calG}}{\epsilon}\; \stackrel{\sim}{\longrightarrow} \; \pview{\mng{\seq{\Sigma}{F_1 : o}}}
\)
which is strong, in that for every $t \in \proj{\trmng{\calG}}{\epsilon}$, we have $t$ and $\phi(t)$ are isomorphic as justified sequences.
Since terminal symbols from $\Sigma$ are assumed to be of order 1, justified sequences from $\proj{\trmng{\calG}}{\epsilon}$ are completely determined by their underlying sequence over $\Sigma$; similarly for $\pview{\mng{\seq{\Sigma}{F_1 : o}}}$.
On the one hand, there is a one-one correspondence between P-views in $\mng{\seq{\Sigma}{F_1 : o}}$ and paths in the generated tree $\mng{\calG}$: given a P-view, the corresponding path is obtained by erasing the O-moves. 
On the other, there is a one-one correspondence between $\proj{\trmng{\calG}}{\epsilon}$ and traversals over $\lambda(\calG)$ projected to symbols from $\Sigma$.
Hence we have the desired set equality.
\end{proof}

\lo{
\subsection{Further applications}
- The approach also works for PCF. Blum's thesis, and Functional Reachability.

- HMOS 2008. Equi-expressivity between HORS and CPDS.
}


\end{document}